%% file: xc-arxiv-2.tex
\documentclass[pdftex]{svjour3}

\usepackage{amsmath,amssymb}
\usepackage[dvips]{graphicx}
\usepackage{subfigure}
\usepackage{fullpage}
\usepackage{enumerate}
\usepackage[dvips]{color}

\numberwithin{theorem}{section}
\numberwithin{equation}{section}
\numberwithin{lemma}{section}
\numberwithin{definition}{section}
\numberwithin{remark}{section}
\numberwithin{example}{section}
\numberwithin{corollary}{section}
\numberwithin{proposition}{section}

\def\mbf#1{\mathchoice{\hbox{\boldmath $\displaystyle #1$}}
        {\hbox{\boldmath $\textstyle #1$}}
        {\hbox{\boldmath $\scriptstyle #1$}}
        {\hbox{\boldmath $\scriptscriptstyle #1$}}}

\newcommand{\X}{{\mbf X}}
\renewcommand{\H}{{\mbf H}}
\newcommand{\bI}{{\mbf I}}

\newcommand{\NTB}{NTB} 
\newcommand{\bINTB}{{\mbf I}_{\text{\NTB}}}
\newcommand{\Y}{{\mbf Y}}
\newcommand{\K}{{\mbf K}}

\newcommand{\Lp}[1][p]{\boldsymbol{L}^{#1}}

\newcommand{\R}{{\mathbb R}}

\newcommand{\E}{{\mathbf E}}

\newcommand{\sB}{\mathcal{B}}

\renewcommand{\P}{\mathbf{P}}

\newcommand{\one}{\mathbf{1}}
\newcommand{\eps}{\varepsilon}
\newcommand{\rhos}{\mathsf{R}}
\newcommand{\rhox}{\mathcal{R}}
\newcommand{\rhot}{\rhox_{\mathrm{\Sigma}}}

\newcommand{\rhonot}{\mathsf{R}_{0}}

\newcommand{\salg}{\mathfrak{F}}

\newcommand{\risk}{r}
\newcommand{\vecrisk}{\mathbf{r}}

\DeclareMathOperator{\cl}{cl}

\DeclareMathOperator{\ES}{ES}
\DeclareMathOperator{\VaR}{VaR}

\renewcommand{\subset}{\subseteq}
\renewcommand{\supset}{\supseteq}

\newlength{\querylen}
\usepackage{fancybox}

\begin{document}

\title{Intragroup transfers, intragroup diversification and their risk
  assessment\footnote{The authors are grateful to Philipp Keller for many interesting and
important remarks on the misuse of IGTs. The comments of the referee
have lead to several improvements of the manuscript. IM was partially supported by
Swiss National Foundation grant 200021\_153597.} }

\author{Andreas Haier \and Ilya Molchanov \and Michael Schmutz}

\institute{A. Haier, I. Molchanov, M. Schmutz \at
  Department of Mathematical Statistics and Actuarial Science,
  University of Bern, Sidlerstrasse 5, 3012 Bern, Switzerland\\
  Tel: +41 31 6318801\\
  Fax: +41 31 631 3870\\
  \email{andreas.haier@gmail.com, ilya.molchanov@stat.unibe.ch, michael.schmutz@stat.unibe.ch}
}

\maketitle

\begin{abstract}
  When assessing group solvency, an important question is to what extent
  intragroup transfers may be taken into account,
  as this determines to which extent diversification can be achieved.
  We suggest a framework to explicitly
  describe the families of admissible transfers that range from the
  free movement of capital to excluding any transactions. The
  constraints on admissible transactions are described as random closed
  sets. The paper focuses on the corresponding solvency tests that
  amount to the existence of acceptable selections of the random
  sets of admissible transactions.


\end{abstract}

\section{Introduction}
\label{sec:introduction}

Risk-based solvency frameworks (such as Solvency~II or the Swiss
Solvency Test~(SST)) assess the financial health of insurance
companies by quantifying the capital adequacy through calculating the
solvency capital requirement. Roughly speaking, companies can use
their own economic capital models (internal models) for calculating
the available capital amounts (net assets) after one year provided the
internal model is approved by the insurance supervisor. The random
variable given by the capital amount after one year is required to be
acceptable with respect to a prescribed risk measure.

Since key market players are organized in groups, the question of
setting appropriate solvency requirements for groups becomes highly
relevant, and the quantification of risks for a group of different
legal entities (agents) is an essential aim of regulators. The
main feature of the group setting is the possibility of intragroup
transfers (IGT) that may alter financial positions of individual
agents. Section~\ref{sec:sol-groups} surveys fundamental ideas in this
relation that have been already mentioned
in~\cite{fil:kun08,fil:kup07,fil:kup08,kel07,lud07} and are
intensively discussed in the financial industry. A key question is
to what extent IGTs should be taken into account for the purpose of
risk assessment.

The choice and admissibility of IGTs may influence the risk
assessment.  These admissible transfer instruments range from the free
movement of capital between the agents (unconstrained approach) to the
case when no transfers are allowed at all (strictly granular risk
assessment). In between, we consider allowing transactions that
prohibit transfers that render the giver bankrupt or those originating
from a bankrupt agent. One might consider imposing some safety margins
or taking into account fungibility issues, see
Section~\ref{sec:restr-transf}.

Our aim is to provide a unified framework in order to explicitly
describe such transfers and the relevant risk measures. So far the
idea of using random closed sets (most importantly random cones) to
describe multiasset portfolios is well established, see
\cite{kab:saf09}. We show that a similar approach may be used to
describe the sets of admissible IGTs. The key idea of our approach is
to regard the group as acceptable if there exists an admissible
transfer that renders acceptable the individual positions of all
agents. The family of vectors representing the capital amounts added
to (or released from) each agent that make the group acceptable serves
as a risk measure for the group. This idea is similar to the risk
assessment of multiasset portfolios from \cite{ham:hey:rud11} and
\cite{cas:mol14}, while the main difference is the non-conical
dependence of the set of admissible transfer instruments from the
current capital position. As a result we end up with the new
possibility to explicitly translate realistic (non-conical) transfer
constraints in a quantitative risk management framework that assesses
genuine intragroup diversification benefits.

Section~\ref{sec:set-valu-portf} recalls major concepts related to
random closed sets, their selections, thereby following and extending
some results from \cite{cas:mol14}. Section~\ref{sec:admissible-crti}
introduces the random sets of admissible IGTs. Furthermore, it
formulates and analyses the relevant capital requirements in terms of
set-valued risk assessment of the relevant random sets of admissible
positions. Since the family of admissible IGTs typically depends to a
large extent on the level of distress a group is faced with, this
dependence creates non-conical and non-linear effects.

Section~\ref{sec:gran-cons-tests} recasts the classical granular and
consolidated approaches into our setting. We show that, under some
canonical assumptions, the granular approach can be obtained by
restricting the set of IGTs in the \emph{unconstrained solvency test},
i.e.\ in the test without any constraints on IGTs. In the coherent
case, the latter corresponds to the consolidated solvency approach,
which is frequently used in practice. Furthermore, we discuss the
unconstrained solvency test in important non-coherent cases. Various
restrictions of transfers are discussed in
Section~\ref{sec:restr-transf}. 

Section~\ref{sec:uneq-plac-agents} deals with the setting of unequally
placed companies, where some companies possess the capital of others
and so simple addition of capital amounts would lead to double
counting. It is shown that this setting can be naturally incorporated
in our set-valued framework.

Separating the financial outcome after a certain period from the
(possibly restricted) transfers between the companies (or in other
applications between portfolios) can lead to diversification effects
that are \emph{different} from the usual case of convex risk measures,
since not all of them work in the same direction, see
Section~\ref{sec:divers-effects}.

In Section~\ref{sec:num-examples} we discuss the computation of the
relevant sets, give inner and outer bounds which can be easily
implemented and discuss some numerical examples.

Our setting differs from studies of systemic risk, where the
individual acceptability of each agent does not suffice for the
acceptability of the whole financial network, see
\cite{amin:fil:min13}. While, in common with the studies of systemic
risks in \cite{fein:rud:web15}, our approach involves inverting a
set-valued risk measure, it does not rely on considering an
equilibrium in the system of agents.

\section{Existing solvency tests for groups}
\label{sec:sol-groups}

\subsection{Legal entity approaches}

A very basic but key observation for market regulators is that it is
not in the obligation of an insurance group but of individual
legal entities to pay for claims of policy holders.
This basic observation can become particularly relevant under stress.
In view of that, it has often been emphasized in the literature that
risk assessment and capital requirements for groups of companies
should take place on an individual basis, see
e.g.~\cite{fil:kun08,fil:kup07,fil:kup08,kel07,lud07}.  This means
that each legal entity is requested to set aside the capital necessary
to make its risk acceptable. This approach is often called \emph{legal
  entity approach}.

The legal entity approach appears in two basic variants:\ stand-alone
and granular. In the \emph{stand-alone} solvency assessment, the
capital amounts (net assets) of each legal entity are modeled separately and
regarded as random variables on a probability space that might differ
between the legal entities. All other group members are considered as
third parties, i.e.\ are treated in the same way as non-members of the
group.

The \emph{granular} approach aims at developing a \emph{joint} model for
all legal entities in a group. The existence of the group has
an impact on the legal entities, meaning that effects of the group on
individual entities should constitute a part of the model. These
effects result from certain IGTs and
ownership relations within the group and should hence be taken into account in the model.
Typical examples are reinsurance agreements, financial guarantees, and
intragroup loans.\footnote{For simplicity of representation we
  consider no hybrid instruments in this paper and we also do not
  enter in the important but non-mathematical debate about legal
  requirements IGTs should
  satisfy.} 
While the whole collection of capital positions
for all involved legal entities is modeled as a random vector, the
granular approach assesses each of its components separately and these
depend on the IGTs chosen by the group. In particular
note that in this \emph{existing} solvency test the capital requirement is
\emph{not} given by a \emph{single figure}.

As opposed to the stand-alone approach, the granular approach relies
on joint modeling of the capital amounts and IGTs so that the information
on other entities and the random variables describing their capital
positions flows directly and consistently\footnote{E.g.\ the terminal
  position of the legal entity usually depends on some risk factors
  (equity, interest rates, mortality etc.). In the stand-alone case
  the risk factor models do not need to be the same.  In the granular
  case the terminal positions are modeled based on the same risk
  factors.} into the modeling of any particular component. In other
words, the granular approach relies on the modeling of each particular
component on a richer probability space.
This may lead to \emph{different marginal distributions} between the
stand-alone and the granular approach. Hence,
it seems to be almost impossible to suitably model on a stand-alone
basis the particularly dangerous situation where several group members
simultaneously run into problems.

Mathematically speaking, consider $d$ legal entities whose terminal
capital positions are described by the vector
$C=(C_1,\dots,C_d)$. Since the main point of this paper is to separate
the financial positions after the relevant time period from the set of
admissible IGTs we will later assume that $C$ is the capital position
of all agents \emph{without} any IGTs.

It is assumed that the risks of the legal entities are evaluated using
monetary risk measures, namely the risk measure $\risk_i$ for the
$i$th legal entity. The individual risks of each legal entity after
IGTs build the vector
\begin{displaymath}
  \vecrisk(\tilde C)=(\risk_1(\tilde C_1),\dots,\risk_d(\tilde C_d))\,,
\end{displaymath}
where the position $\tilde C$ represents the capital positions after
IGTs. It is regarded acceptable if $\risk_i(\tilde C_i)\leq 0$ for
all $i$, meaning that $\vecrisk(\tilde C)\leq 0$, where the inequality
between vectors is understood coordinatewisely and so means that
$\vecrisk(\tilde C)$ has all non-positive components.

\subsection{Consolidated approach vs.\  granular with fixed IGTs}
\label{sec:consolidated-fungibility}

Intragroup transactions can be used for increasing the
``diversification'' of risks within the group, and so they \emph{can}
be in the interest of all policy holders, see e.g.~\cite{fil:kup08},
if they are applied in \emph{reality}. In particular, IGTs may be
used to offset the risks of some, e.g.\ poorly performing, legal
entities, while not necessarily immediately diminishing the disposable
assets of other legal entities.

Since the IGTs can reduce the solvency capital requirements, it is
important that they are not of a purely hypothetical nature. Thus,
they have to be \emph{realistic} in situations when they are
needed. The family of feasible transfers depends on whether or not the
group is in a stressed situation, i.e.\ they are particularly exposed
to \emph{fungibility risks}, as is well-known in
practice. Furthermore, note that not all policy holders (of different
legal entities within a group) have the same interests.  Hence, IGTs
have to be ``sufficiently balanced'' with respect to the interests of
all \emph{policy holders}.
It is important to note that IGTs have a general potential for being
misused. Thus, even if transfers are based on legally binding and
enforceable contracts, it is still possible that the corresponding
transfers are not realistic, not sufficiently in line with the
interests of some policy holders, etc.,
so that other/further restrictions may have to be taken into account.
The decision about the restrictions, which are finally taken into
account in a concrete solvency framework is a political and legal but
not mathematical question, therefore it is not addressed in this
paper. However, what we provide here is a tool to translate
restrictions into a quantitative risk management framework.

It becomes increasingly popular in practice to regard the group as
acceptable if the random variable
\begin{equation}
  \label{eq:def-D}
   D=\sum_{j=1}^dC_j
\end{equation}
is acceptable with respect to a prescribed risk measure.\footnote{For
  simplicity of the representation we ignore here subtleties regarding
  the so-called (a)symmetric valuation. For a discussion of
  participation in subsidiaries we refer to
  Section~\ref{sec:uneq-plac-agents}.}  This approach is called the
\emph{consolidated solvency test}.  It is a frequently discussed
interpretation that the classical consolidated approach
\emph{implicitly} assumes full fungibility of capital between all
different legal entities of the group, see e.g.~\cite{kel07,lud07}.

Filipovi{\'c} and Kupper \cite{fil:kup07,fil:kup08} and Filipovi{\'c}
and Kunz \cite{fil:kun08} made an important step towards understanding
intra group diversification and quantifying regulatory schemes being
sandwiched between the granular approach without any admissible IGTs,
in the following called \emph{strictly granular}, and the consolidated
solvency tests. Concrete examples are based on dividend payments and
reinsurance contracts and are described by some random transfer
instruments (linearly independent random variables) $Z_0,\dots,Z_n$,
so that the terminal risk profile of each agent is given by
\begin{displaymath}
  C_i+\sum_{j=0}^n x_i^j Z_j\,,\quad i=1,\dots,d\,,
\end{displaymath}
where the transferred amounts $x_0,\dots,x_n$ belong to a specified
(feasible) subset of $\R^{n+1}$ and satisfy the clearing condition
\begin{displaymath}
  \sum_{i=1}^d\sum_{j=0}^n x_i^j Z_j\leq 0 \quad \text{a.s.}
\end{displaymath}
It is assumed in~\cite{fil:kup08} that the group aims to minimise
the aggregate required group capital
\begin{displaymath}
  \sum_{i=1}^d \risk_i\Big(C_i+\sum_{j=0}^n x_i^j Z_j\Big)\,,
\end{displaymath}
subject to the feasibility and clearing conditions on the weights
$x_i^j$.

Apart from the mentioned independency assumption, \cite{fil:kup08}
assume a single currency setting, absence of transaction costs, that
the weights are deterministic numbers and that the admissible
transfers are restricted to be a linear combination of fixed 
instruments without explicit fungibility constraints that depends on
the level of distress a group is exposed to.

Fungibility constraints on possible dividend payments are considered
in the very concrete bottom-up group diversification analysis
presented in~\cite{fil:kun08}.  The natural link between admissible
IGTs and the minimization problem of the total required capital of
the group, which is emphasized in~\cite{fil:kup08}, relates this work
to the extensive literature on \emph{optimal risk sharing}, see
e.g.~\cite{acc07,acc:svi09,Asim:Bad:Tsan13,bar:elk05,bar:sca08,ger79,rav:svi14}
and the literature cited therein. For optimization under restriction
to certain so-called cash invariant sets we refer to~\cite{fil:svi08},
partially based on~\cite{fil:kup08}, for portfolios of risk vectors
(including the influence of dependence on the risk of a portfolio)
see~\cite{kis:rue10}.

In the interest of brevity, the main focus of this paper concerning
the granular approach is on the task to directly include concrete
admissibility constraints for IGTs in a solvency framework without
restricting the analysis to concretely given contracts, and we leave to
agents the task of choosing specific IGTs from the family of
admissible ones.

From a practical perspective, groups may not necessarily aim to
minimize the total capital requirement, since some legal entities may
have different placements or roles within the group (like subsidiaries
and head offices), some may be reluctant to commit own capital in
order to compensate the losses of other legal entities or may do this
only given certain conditions (that are random), etc. In view of this,
the total required capital for the group may not serve as a right
utility for the group, or at least the agents might seek to minimize
the total required capital only under additional constraints that do not
seem to be reflected in the literature so far.

\subsection{Consolidated test with fungibility constraints}

The consolidated solvency approach is based on the acceptability of
the random variable in~(\ref{eq:def-D}), where the group is simply
considered as if it were one ``legal entity''. It is the default case
in Solvency~II.  Since July 2015 the default case of the SST is also
based on this approach. However the regulator may impose additional
requirements concerning availability and fungibility of capital within
the group and also impose capital add-ons in case of seriously
restricted fungibility that is not reflected in the models.
In view of that, it is essential to derive a proper and mathematically
well-founded way of describing transfer possibilities within the group
and to find and to derive a framework which allows to
\emph{adequately} include fungibility constraints into the
consolidated framework.

By translating the consolidated test into the unconstrained solvency
test in the subadditive case, we show that no fungibility constraints
are taken into account in the risk-measurement.
We also suggest a framework that represents fungibility constraints in
a solvency test via simply and explicitly restricting the set of
admissible transfers.

\section{Set-valued portfolios and set-valued risks}
\label{sec:set-valu-portf}

A set $A\subset \R^d$ is \emph{lower} if $y\leq x$ (coordinatewisely)
for $x\in A$ implies that $y\in A$. The family of upper sets coincides
with the family of reflected lower sets, i.e.\ $A$ is an upper set if
and only if $\{-x:\; x\in A\}$ is a lower set. The
topological closure of $A\subset\R^d$ is denoted by $\cl A$ and its
boundary by $\partial A$. Denote $A+a=\{x+a:\; x\in A\}$.

Fix a complete probability space $(\Omega,\salg,\P)$.  Let $\X$ be a
lower \emph{random closed set} in $\R^d$, i.e.\ $\X$ is a random
element taking values in the family of lower closed sets in $\R^d$.
The measurability requirement on $\X$ is understood as $\{\omega:\;
\X(\omega)\cap K\neq\emptyset\}\in\salg$ for all compact sets $K$ in
$\R^d$, see \cite{mo1}.  The set $\X$ is called a \emph{set-valued
  portfolio} in \cite{cas:mol14}. In our setting, the points of $\X$
describe the terminal capitals of $d$ companies of a group after all
admissible IGTs. In many cases, $\X$ is almost surely convex, meaning
that almost all its realizations are convex sets.

Let $\Lp(\R^d)$ be the family of $p$-integrable random vectors in
$\R^d$, where $p\in\{0\}\cup[1,\infty]$ is fixed. The choice $p=0$
yields the family of all random vectors.  Denote by $\|\xi\|_p$ the
$\Lp$-norm of $\xi$ for $p\geq1$.  A random vector $\xi$ in $\R^d$ is
said to be a \emph{selection} of $\X$ if $\xi\in\X$ almost
surely. Such a random vector may be viewed as a particular terminal
position achieved after a certain IGT. We assume throughout that $\X$
is \emph{$p$-integrable}, i.e.\ $\X$ possesses at least one
$p$-integrable selection. In other words, the family $\Lp(\X)$ of all
$p$-integrable selections of $\X$ is not empty.

We identify $\X$ with the family $\Lp(\X)$ of all its $p$-integrable
selections. This is justified by the fact that, for almost all
$\omega\in\Omega$, $\X(\omega)$ is the closure of
$\{\xi_n(\omega),n\geq1\}$ for a sequence $\{\xi_n,n\geq1\}\subset
\Lp(\X)$, see \cite[Prop.~2.1.2]{mo1}.

In the following $\vecrisk=(\risk_1,\dots,\risk_d)$ denotes the vector
composed of \emph{monetary risk measures} with \emph{finite} values
defined on the space $\Lp(\R)$ for $p\in\{0\}\cup[1,\infty]$ (called
$\Lp$-risk measures).  Canonical examples of such risk measures are
the Value-at-Risk for all $p$, the Average Value-at-Risk for $p\geq 1$
or the entropic risk measure for $p=\infty$.

Being finite $\Lp$-risk measures, the components of $\vecrisk$ are
Lipschitz in the $\Lp$-norm, see \cite{kain:rues09}. In particular,
they are strongly continuous. If $p=\infty$, additionally assume that
the components of $\vecrisk$ satisfy the \emph{Fatou property} that
corresponds to the weak-star lower semicontinuity (that is with
respect to the bounded a.s.\ convergence).

Furthermore, $\vecrisk$ is said to be \emph{coherent}
(resp.\ convex) if all its components are coherent (resp. convex) risk
measures, and in this case we let $p\in[1,\infty]$, see
\cite{kain:rues09}.  The coherency or convexity assumptions are
explicitly imposed whenever needed.
\begin{quote}
  All convex (and coherent) risk measures are tacitly assumed to be
  \emph{law invariant} and defined on a \emph{non-atomic} probability
  space.
\end{quote}

\begin{definition}
  A set-valued portfolio $\X$ is said to be \emph{acceptable}, if it
  possesses a $p$-integrable selection with all individually
  acceptable marginals, i.e.\ there exists $\xi\in \Lp(\X)$ such that
  $\vecrisk(\xi)\leq 0$. Then $\xi$ is called an \emph{acceptable
    selection} of $\X$.
\end{definition}

\begin{definition}
  \label{def:rhos}
  The \emph{selection risk measure} $\rhos(\X)$ is the closure of the
  set
  \begin{displaymath}
    \rhonot(\X)=\{a\in\R^d:\; \X+a\; \text{is acceptable}\}\,.
  \end{displaymath}
\end{definition}

Equivalently, the selection risk measure can be defined as
\begin{displaymath}
  \rhos(\X)=\mathrm{cl}\bigcup_{\xi\in \Lp(\X)}(\vecrisk(\xi)+\R_+^d).
\end{displaymath}
In \cite{ham:rud:yan13}, $\rhos(\xi+\K)$ for $\xi\in\Lp(\R^d)$ and a
cone $\K$ is called a market extension of the regulator risk measure
$\vecrisk$.

\begin{example}
  On the line all lower sets are half-lines, so that each set-valued
  portfolio in $\R^1$ is given by $\X=(-\infty,\eta]$ for a random
  variable $\eta$. Then $\X$ is acceptable if and only if $\eta$ is
  acceptable. Working with the half-line $\X$ instead of $\eta$ does
  not alter financial realities, while being a useful tool in
  higher-dimensional situations.

  If $\X=\xi+\R_-^d$ for $p$-integrable random vector $\xi$, then
  $\rhos(\X)=\vecrisk(\xi)+\R_+^d$.
\end{example}

\begin{example}
  \label{ex:half-space}
  Let $\X=\{x=(x_1,\dots,x_d)\in\R^d:\; x_1+\cdots+x_d\leq \eta\}$ for
  a random variable $\eta$. The family $\Lp[\infty](\partial \X)$ of
  all essentially bounded selections of the boundary of $\X$ was
  studied in detail in \cite{jouin:sch:touz08} for $d=2$ and is called
  the set of attainable allocations, see also \cite{ger79}.
\end{example}

In order to find acceptable selections of $\X$, it is sensible to look
only at those points of $\X$ that are not coordinatewisely dominated
by any other selection of $\X$. These points build a subset of
$\partial\X$ denoted by $\partial^+\X$ and are called \emph{Pareto
  optimal} points of $\X$.

\begin{lemma}
  \label{lemma:pareto-convex-colsed}
  If $\X$ is convex, then $\partial^+\X$ is a random closed set.
\end{lemma}
\begin{proof}
  Assume $x_n=(x_n^{(1)},\dots,x_n^{(d)})\in \partial^+\X$ is a
  sequence converging to $x$. By choosing subsequences, we can assume
  that all components converge monotonically. Let $T$ be the set of
  components converging strongly decreasing.  Assume $x
  \notin \partial^+\X$. Since $\X$ is closed there exists
  $y\in\partial^+\X$ such that $x\leq y$ and $y^{(i)}>x^{(i)}$ for
  some $i$. Choose $y$ such that the set $S$ of all indices $i$ for
  which this inequality holds is maximal. If $T\subseteq S$, then for
  sufficiently large $n$, $y$ dominates $x_n$, contradicting the
  Pareto optimality of $x_n$.  Assume that $j\in T\backslash S$. By
  convexity, $\tilde{y}_n=\lambda y + (1-\lambda) x_n\in \X$ for
  $\lambda\in[0,1]$. By taking $\lambda$ sufficiently close to $1$, we
  can achieve that $\tilde{y}_n^{(i)}>x^{(i)}$ for $i\in S$ and for
  sufficiently large $n$. Due to the strict monotonicity of
  $x_n^{(j)}$, we also have $\tilde{y}_n^{(j)}>x^{(j)}$, contradicting
  the maximality of $S$. Thus, $\partial^+\X$ is closed.

  For the measurability of $\partial^+\X$, it suffices to check that
  $\Gamma=\{(\omega,x):\; x\in \partial^+\X(\omega)\}$,
  i.e.\ the graph of $\partial^+\X$, is a measurable set in the
  $\sigma$-algebra $\salg\otimes\sB(\R^d)$, where $\sB(\R^d)$ is the
  Borel $\sigma$-algebra in $\R^d$. Indeed,
  \begin{displaymath}
    \Gamma=\bigcap_{q\in{\mathbb Q}_+^d}
    \{(\omega,x):\; x\in\X(\omega), x+q\notin\X\}\,,
  \end{displaymath}
  where ${\mathbb Q}_+$ is the family of positive rational
  numbers. This is justified, since a convex lower set is necessarily
  regular closed, i.e.\ coincides with the closure of its
  interior.
\end{proof}

\begin{lemma}
  \label{lemma:pareto-x}
  A convex set-valued portfolio $\X$ admits an acceptable selection if
  and only if $\partial^+\X$ admits an acceptable selection.
\end{lemma}
\begin{proof}
  Assume that $\X$ admits an acceptable selection $\xi$. If $\xi$ is
  not Pareto optimal, consider the random closed set
  $\Y=\X\cap(\xi+\R_+^d)$. All selections of $\Y$ are acceptable,
  $\Y\cap\partial^+\X$ is almost surely non-empty and has a measurable
  graph. By the measurable selection theorem
  \cite[Th.~5.4.1]{kab:saf09}, $\Y\cap \partial^+\X$ admits a
  measurable selection that is automatically acceptable.
\end{proof}

The scaling transformation of a set-valued portfolio is defined as
$t\X=\{tx:\; x\in\X\}$.  The sum of set-valued portfolios $\X+\Y$ is
the set-valued portfolio being the closure of all sums of selections
of $\X$ and $\Y$. It is known \cite{mo1} that such operations respect
the measurability property, i.e. $t\X$ and $\X+\Y$ are random closed
sets.

The following result is proved in \cite{cas:mol14} for $\vecrisk$
composed of coherent risk measures, while obvious changes lead to its
version for general monetary risk measures.

\begin{theorem}
  \label{thr:selection-risk-monetary}
  The selection risk measure takes values being upper closed sets, and
  also
  \begin{itemize}
  \item[(i)] $\rhos(\X+a)=\rhos(\X)-a$ for all deterministic
    $a\in\R^d$ (cash invariance);
  \item[(ii)] If $\X\subset\Y$ a.s., then $\rhos(\X)\subset\rhos(\Y)$
    (monotonicity).
  \end{itemize}
  If $\vecrisk$ is convex and $\X,\Y$ are almost surely convex
  set-valued portfolios, then $\rhos(\X)$ takes convex values, is law
  invariant, and
  \begin{itemize}
  \item[(iii)] $\rhos(\lambda\X+(1-\lambda)\Y)\supset
    \lambda\rhos(\X)+(1-\lambda)\rhos(\Y)$ for all deterministic
    $\lambda\in[0,1]$ (convexity).
  \end{itemize}
  If, additionally, the components of $\vecrisk$ are all homogeneous
  (i.e.\ $\vecrisk$ is coherent), then
  \begin{itemize}
  \item[(iv)] $\rhos(t\X)=t\rhos(\X)$ for all $t>0$ (homogeneity);
  \item[(v)] $\rhos(\X+\Y)\supset \rhos(\X)+\rhos(\Y)$,
  \end{itemize}
  meaning that $\rhos$ is a set-valued coherent risk measure, see
  \cite{ham:hey10,ham:hey:rud11}.
\end{theorem}

The set $\partial^+\X$ is said to be \emph{$p$-integrably bounded} if
\begin{displaymath}
  \|\partial^+\X\|=\sup\{\|x\|:\; x\in\partial^+\X\}\in \Lp(\R).
\end{displaymath}
If $p=\infty$, this is the case if and only if $\partial^+\X$ is
almost surely a subset of a deterministic bounded set.  For the sake
of completeness we provide a proof of the closedness of $\rhonot(\X)$
for all $p$ that does not use the coherency assumption as in
\cite[Th.~3.6]{cas:mol14}.

\begin{proposition}
  \label{prop:closed}
  If $\partial^+\X$ is $p$-integrably bounded with $p\in[1,\infty]$
  and $\vecrisk$ is a convex $\Lp$-risk measure, then $\rhonot(\X)$ is
  a closed set.
\end{proposition}
\begin{proof}
  Let $x_n\in\rhonot(\X)$ and $x_n\to x$.  By
  Lemma~\ref{lemma:pareto-x}, there exists $\xi_n\in
  \Lp(\partial^+\X)$ such that $\vecrisk(\xi_n)\leq x_n$. 

  Assume first that $p\in[1,\infty)$. Since $\partial^+\X$ is
  $p$-integrably bounded, all its selections have uniformly bounded
  $\Lp[1]$-norms. By the Koml\'os theorem, see
  e.g. \cite[Th.~5.2.1]{kab:saf09}, and passing to a subsequence,
  $\bar\xi_n=n^{-1}(\xi_1+\cdots+\xi_n)$ converges a.s. to $\xi$. Then
  $\bar\xi_n$ almost surely belongs to the convex hull of
  $\partial^+\X$, whence $\|\bar\xi_n\|\leq \|\partial^+\X\|$ and
  $\|\xi\|\leq\|\partial^+\X\|$. Thus, $\bar\xi_n\to\xi$ in $\Lp$. The
  $\Lp$-continuity of the components of $\vecrisk$ yields that
  \begin{displaymath}
    \vecrisk(\xi)=\lim \vecrisk(\bar\xi_n)
    \leq \lim n^{-1}(x_1+\cdots+x_n)=x,
  \end{displaymath}
  so that $x\in\rhonot(\X)$.

  If $p=\infty$, the above inequality also applies in view of the
  assumed Fatou property and the fact that the norms of $\bar\xi_n$
  are all bounded by the essential supremum of $\|\partial^+\X\|$.
\end{proof}

Denote by
\begin{displaymath}
  h_\X(u)=\sup\{\langle u,x\rangle:\; x\in\X\}
\end{displaymath}
the \emph{support function} of $\X$, where $\langle\cdot,\cdot\rangle$ is the
scalar product in $\R^d$. Then $h_\X(u)$ is a random variable for each
$u$ that may take infinite values. The following result provides a
simple \emph{outer bound} for the selection risk measure of $\X$.

\begin{theorem}[see Prop.~4.6 \cite{cas:mol14}]
  \label{cor:lower-bound}
  Assume that $\vecrisk=(\risk,\dots,\risk)$ has all identical
  components for a coherent $\Lp$-risk measure $\risk$. Then
  \begin{equation}
    \label{eq:upper-bound}
    \rhos(\X)\subset \bigcap_{u\in\R_+^d} \Big\{x:\; \langle x,u\rangle
    \geq \risk(h_{\X}(u))\Big\},
  \end{equation}
  where $\risk(h_{\X}(u))=-\infty$ if $h_\X(u)=\infty$ with a positive
  probability.
\end{theorem}

\section{Admissible IGTs, attainable positions and their risks}
\label{sec:admissible-crti}

\subsection{Admissible IGTs}

Recall that $C=(C_1,\dots,C_d)$ denotes the terminal positions of the
legal entities evaluated on the granular basis, all expressed in the
same currency. Assume that $C$ is $p$-integrable.

A family of \emph{admissible IGTs} is identified as the family
$\Lp(\bI)$ of $p$-integrable selections of a random closed set $\bI$
in $\R^d$. It is often the case that $\bI$ depends on the terminal
capital positions $C$ and in this case $\bI=\bI(C)$ is written as a
function of $C$ that might also depend on additional randomness,
e.g. random exchange rates.  This gives the possibility to model the
important feature that realistic transfer possibilities depend on the
level of distress of the economic environment.

The \emph{attainable financial positions} at the terminal time after
admissible transfers form the family of selections of the random
closed set
\begin{displaymath}
  \X(C)=C+\bI(C).
\end{displaymath}
It is natural to regard the set $\X=\X(C)$ of attainable positions
preferable over another set $\Y=\Y(C)$ if, for each selection
$\eta\in\Y$, there is a selection $\xi\in\X$ such that $\eta\leq\xi$
with probability one. This partial order can be realized as the
inclusion order $\Y\subset\X$ if the sets of attainable positions are
lower sets in $\R^d$.  For this, we assume that with each admissible
IGT given by a random vector $\zeta$, the set $\bI(C)$ also includes
points that are less than or equal to $\zeta$ in the coordinatewise
order, so that $\bI(C)$ and $\X(C)$ are lower sets. The lower set
assumption is useful to formulate mathematical properties of
risks. While initially it might not seem reasonable to consider IGTs
that involve a disposal of some of the assets, the monotonicity
property of risk measures implies that the agents or their group in no
circumstances would opt for an IGT that involves uncompensated
disposal of assets and even if they would pursue such IGT, then the
position without such a disposal is also acceptable.

We assume throughout the rest of the paper that $\bI$ is a lower set
that almost surely contains the origin, and
\begin{equation}
  \label{eq:subset-H}
  \R^d_-\subset \bI(C)\subset \H=\Big\{x\in\R^d:\; \sum x_i\leq 0\Big\}.
\end{equation}
It means that the nil-transfer is admissible and that admissible
intragroup transfers are financed by the group.  Sometimes is is
useful to assume also that
\begin{equation}
  \label{eq:iteration}
  \bI(C+y)\subset \bI(C)-y
\end{equation}
for each $y\in \bI(C)$. Equivalently, $\X(C')\subset \X(C)$ for all
$C'\in\X(C)$, meaning that the result cannot be improved by
substituting one large transaction by several small ones.

It is essential to stress that $\bI(C)$ is not necessarily a cone. 
In many examples, the set $\bI(C)$ is convex, but it is not
necessarily the case, e.g.\ for fixed transaction costs and
indivisible assets.

\begin{example}
  If there is a fixed range of admissible IGTs given by
  $\{x^{(1)},\dots,x^{(k)}\}$, then generally
  \begin{displaymath}
    \bI(C)=\bigcup_{i=1}^k \big(x^{(i)}+\R_-^d\big)
  \end{displaymath}
  is a non-convex set that does not depend on $C$.
\end{example}

\subsection{Risks of a group}
\label{sec:total-risk}

The position $C$ together with the corresponding admissible IGTs given
by $\bI(C)$ (or the corresponding set $\X(C)$ of attainable positions)
is acceptable if $0\in\rhos(\X(C))$.  The conventional definition of
risk measures in its set-valued variant \cite{ham:hey10,cas:mol14}
suggests passing from the acceptability criterion to the risk measure
by considering the set of all $x\in\R^d$ such that $\X(C+x)$ is
acceptable.

\begin{definition}
  \label{def:group-risk}
  The \emph{group risk} associated with the attainability set $\X(C)$
  is
  \begin{equation}
    \label{eq:rhox}
    \rhox(\X(\cdot),C)=\big\{x\in\R^d:\; 0\in\rhos(\X(C+x))\big\}.
  \end{equation}
\end{definition}

\begin{remark}
  \label{rem:inverse}
  The group risk can be regarded as the inverse of the set-valued
  function $x\mapsto \rhos(\X(C+x))$, see \cite{aub:fra90}. A similar
  inverse appears in \cite{fein:rud:web15} as an approach sensitive to
  the capital levels, where $C$ denotes the set of capital amounts for
  agents, $\X(C)$ is the set of equilibrium prices, and the inverse of
  the selection risk measure of $\X(C+x)$ (in our notation) determines
  the systemic risk associated with the system of
  agents. Definition~\ref{def:group-risk} can be applied to determine
  the risks of some multiasset portfolios from
  \cite[Sec.~2.3]{cas:mol14} that depend non-linearly on the financial
  position. Note that the first argument of $\rhox$ in~(\ref{eq:rhox})
  is a function.
\end{remark}

\begin{remark}
  \label{rem:cone-X}
  If $\bI(C)=\bI$ is a convex cone that does not depend on $C$, like
  it is the case for the conical model of proportional transaction
  costs (see \cite{ham:hey:rud11,kab:saf09,cas:mol14}), then
  $\X(C+x)=\X(C)+x$, so that $\rhos(\X(C+x))=\rhos(\X(C))-x$, whence
  $\rhox(\X(\cdot),C)=\rhos(\X(C))$ is a convex set.  As we see later
  on, in many cases of assessing the group risk, the set $\bI(C)$
  depends on $C$, so that $\X(C+x)$ may substantially differ from
  $\X(C)+x$. Then $\rhox(\X(\cdot),C)$ may become non-convex and so
  considerably more complicated to compute.
\end{remark}

\begin{proposition}$ $
  \label{prop:tot-risk-mon}
  \begin{enumerate}[(i)]
  \item If $\X(x)\subset \Y(x)$ for all $x\in\R^d$, then
    $\rhox(\X(\cdot),C)\subset\rhox(\Y(\cdot),C)$.
  \item $\rhox(\X(\cdot),C)$ contains $\vecrisk(C)+\R_+^d$.
  \item If \eqref{eq:iteration} holds, then $\rhox(\X(\cdot),C)$ is an
    upper set and is non-decreasing as function of $C$, that is
    $\rhox(\X(\cdot),C)\subset\rhox(\X(\cdot),C')$ if $C\leq C'$ a.s.
  \end{enumerate}
\end{proposition}
\begin{proof}
  \textsl{(i)} In this case $\rhos(\X(C+x))\subset \rhos(\Y(C+x))$ and
  so the inverse function given by \eqref{eq:rhox} is also monotone. 

  \textsl{(ii)} Since $\R^d_-\subset \bI(C)$, we have
  \begin{displaymath}
    \rhos(\X(C))\supset \rhos(C+\R^d_-)=\vecrisk(C)+\R^d_+,
  \end{displaymath}
  so that $\rhox(\X(\cdot),C)\supset(\vecrisk(C)+\R^d_+)$ by (i). 

  \textsl{(iii)} Assume that $C\leq C'$, so that $y=C-C'\in\R^d_-$. By
  \eqref{eq:subset-H}, $z\in\bI(C)$, whence \eqref{eq:iteration}
  yields that $\bI(C)\subset\bI(C')-C+C'$. Thus, $\X(C)\subset\X(C')$
  and $\rhos(\X(C+x))\subset \rhos(\X(C'+x))$, and the total risk is
  monotonic by (i).

  If $x\leq y$, then the above applies to $C+x$ instead of $C$ and
  $C+y$ instead of $C'$, so that
  $\rhos(\X(C+x))\subset\rhos(\X(C+y))$, whence if
  $0\in\rhos(\X(C+x))$, then also $0\in\rhos(\X(C+y))$.
\end{proof}

If all agents operate with the same currency, it is possible to
quantify the risk using a single real number by considering the minimal total
capital requirement for the group.

\begin{definition}
  The \emph{total risk} associated with $\X(C)$ (also called the
  total group risk) is defined by
  \begin{equation}
    \label{eq:total-risk}
    \rhot(\X(\cdot),C)
    =\inf\left\{\sum_{i=1}^d x_i:\; 0\in\rhos(\X(C+(x_1,\dots,x_d)))\right\}\,.
  \end{equation}
\end{definition}

The total risk does not change if $C$ is replaced by $C+z$ for a
deterministic vector $z$ with $\sum z_i=0$. By
Proposition~\ref{prop:tot-risk-mon},
\begin{displaymath}
  \rhot(\X(\cdot),C)\leq \sum \risk_i(C_i).
\end{displaymath}
It is easy to see that the total risk is the support function of
$\rhox(\X(\cdot),C)$ in direction $(-1,\dots,-1)$. In particular, the
acceptability of $\X(C)$ yields that $\rhot(\X(\cdot),C)\leq 0$, but the
opposite conclusion is not necessarily true. The non-positivity of the
total risk yields only the existence of transfers $(x_1,\dots,x_d)$
with the total capital requirement being zero that make $\X(C+x)$
acceptable.  If the infimum in \eqref{eq:total-risk} is attained, then
the vectors $x=(x_1,\dots,x_d)$ that provide the infimum give possible
allocations of the total risk between the legal entities. Then there
is an acceptable selection $\xi$ of $\X(C+x)$, and the regulator could
possibly request \emph{conclusion of legally binding contracts} for
transfers in order to arrive from $C$ to $\xi$.

The set-valued map $\X(C+x)$ is said to be
\emph{upper semicontinuous} as function of $x$ if, for all $\eps>0$,
$x\in\R^d$, and any sequence $x_n$ that converges to $x$,
\begin{equation}
  \label{eq:up-semi}
  \X(C+x_n)\subset \X(C+x)+B_{\zeta_n},
\end{equation}
where $B_{\zeta_n}$ is the closed ball of
radius $\zeta_n$ centred at the origin and $\|\zeta_n\|_p\to0$. This
property can be equivalently formulated for $\bI(C+x)$. 

\begin{proposition}
  \label{prop:closed-total}
  Let $\vecrisk$ be a coherent $\Lp$-risk measure with
  $p\in[1,\infty]$.  If $\partial^+\X(C+x)$ is $p$-integrably bounded
  for all $x\in\R^d$ and $\X(C+x)$ is upper semicontinuous as function
  of $x$, then the set $\rhox(\X(\cdot),C)$ is closed. If also
  $\X(C+x)\subset \xi+x+\R_-^d$ for at least one $\xi\in\Lp(\R^d)$ and all
  $x\in\R^d$, then the infimum in \eqref{eq:total-risk} is attained.
\end{proposition}
\begin{proof}
  For each $x\in\R^d$, the set $M(x)=\rhos(\X(C+x))$ is closed by
  Proposition~\ref{prop:closed}. Assume that $0\in M(x_n)$, $n\geq1$, and
  $x_n\to x$. By Theorem~\ref{thr:selection-risk-monetary}(ii),
  \begin{displaymath}
    M(x_n)\subset \rhos(\X(C+x)+B_{\zeta_n}).
  \end{displaymath}
  For each selection $\xi$ of $\X(C+x)+B_{\zeta_n}$, there exists a
  selection $\xi'$ of $\X(C+x)$ such that $\|\xi-\xi'\|\leq \zeta_n$.
  Since the components of $\vecrisk$ are Lipschitz in the $\Lp$-norm
  (see \cite{kain:rues09}),
  $\|\vecrisk(\xi)-\vecrisk(\xi')\|\leq c\|\zeta_n\|_p$ for a constant
  $c$, so that
  \begin{displaymath}
    M(x_n)\subset M(x)+B_{\eps_n}
  \end{displaymath}
  for $\eps_n=c\|\zeta_n\|_p\to0$. Thus, $0\in M(x)+B_{\eps_n}$ for
  all $n$. In view of the closedness of $M(x)$, it contains the
  origin, so that $x\in\rhox(\X(\cdot),C)$.

  The monotonicity of the group risk (see
  Proposition~\ref{prop:tot-risk-mon}) yields that
  $\rhox(\X(\cdot),C)\subset \vecrisk(\xi)+\R_+^d$, and so the
  attainability of the infimum follows.
\end{proof}

The following basic properties of the introduced risks are easy to
prove.

\begin{proposition}
  \label{prop:total-properties}$ $
  \begin{enumerate}[(i)]
  \item The group risk and the total risk are cash invariant, i.e.
    \begin{align*}
      \rhox(\X(\cdot+a),C+a)&=\rhox(\X(\cdot),C)-a,\\
      \rhot(\X(\cdot+a),C+a)&=\rhot(\X(\cdot),C)-\sum_{i=1}^d a_i.
    \end{align*}
  \item If $\vecrisk$ used to construct the selection risk measure is
    a homogeneous risk measure and $\bI(tC)=t\bI(C)$ for all $t>0$,
    then $\rhox(\X(t\cdot),tC)=t\rhox(\X(\cdot),C)$ and
    $\rhot(\X(t\cdot),tC)=t\rhot(\X(\cdot),C)$ for all $t>0$.
  \item If \eqref{eq:iteration} holds, then
    $\rhot(\X(\cdot),C)\leq\rhot(\X(\cdot),C')$ for $C'\leq C$.
  \end{enumerate}
\end{proposition}

\begin{remark}
  Despite the group risk has natural properties of a monetary
  set-valued risk measure and the total risk is similar to monetary
  risk measures, we avoid calling them risk measures, since they depend
  on two arguments: a random closed set $\X(C)$ and a specific random point
  $C$ inside this set.
\end{remark}

It is known \cite{lep:mol16} that risk assessment for multiasset
models with random exchange rates may be subject to the so-called
\emph{risk arbitrage} meaning that it is possible to find a sequence
of selections that can be made acceptable by adding capital that
tends to minus infinity, so that it is possible to release an infinite
capital maintaining the acceptability of the position. This is the
case if and only if the total risk attains the value $-\infty$. We
will show that this is impossible for convex risk measures due to condition
\eqref{eq:subset-H}. However, it may become a relevant issue in the
multi-currency setting with random exchange rates, see
Section~\ref{sec:transaction-costs}, and for general monetary risk
measures. Recall that inf-convolution of coherent risk measures is
defined by taking the closed convex hull of their acceptance sets, see
\cite{delb12}. 

\begin{proposition}
  If $\vecrisk$ is a convex risk measure with all identical components
  or a coherent risk measure with a non-trivial inf-convolution of its
  components, then the corresponding total risk
  is different from $-\infty$ for all sets of IGTs that
  satisfy \eqref{eq:subset-H}.
\end{proposition}
\begin{proof}
  Condition~\eqref{eq:subset-H} yields that $\X(C)\subset C+\H$, so
  that $\rhox(\X(\cdot),C)\subset
  \rhox(\cdot+\H,C)=\rhos(C+\H)$. Thus, it suffices to consider the
  latter set. Assume first that $\vecrisk$ is coherent and there
  exists a sequence $\{x^{(k)},k\geq1\}$ in $\R^d$ such that
  $\sum_{i=1}^d x_i^{(k)}\to-\infty$ as $k\to\infty$ and
  $0\in\rhos(C+\H+x^{(k)})$ for all $k$. Therefore, there exist
  $\xi^{(k)}\in\Lp(\R^d)$, $k\geq1$, such that $\sum_{i=1}^d
  \xi_i^{(k)}\leq0$ a.s. and $\vecrisk(C+\xi^{(k)}+x^{(k)})\leq
  0$. Denote by $\risk_*$ the inf-convolution of the components of
  $\vecrisk$. Then $\risk_*(C_i+\xi_i^{(k)}+x_i^{(k)})\leq0$ for all
  $i=1,\dots,d$. By subadditivity,
  \begin{displaymath}
    \risk_*\Big(\sum_{i=1}^d C_i+\sum_{i=1}^d x_i^{(k)}\Big)
    \leq \risk_*\Big(\sum_{i=1}^d C_i+\sum_{i=1}^d\xi_i^{(k)}+\sum_{i=1}^d
    x_i^{(k)}\Big)
    \leq 0.
  \end{displaymath}
  Thus, $\risk_*(\sum C_i)=-\infty$, which is excluded by the
  condition.  In the convex case the proof is similar with
  $\risk_*=\risk$ and the sums
  replaced by convex combinations with weights $1/d$.
\end{proof}

The following example shows that using non-convex risk measures (like
the Value-at-Risk) may lead to risk arbitrage in the high-dimensional
setting. This kind of arbitrage is intimately related to the notions
of divisibility of risk measures as can be seen when comparing the
following example to the proof of Proposition~2.2 in~\cite{wang14}.
Throughout the paper we use the following definition of the Value-at-Risk
\begin{equation}
  \label{eq:def-var}
  \VaR_\alpha(\eta)=-\inf\big\{x:\;\P(\eta\leq x)\geq \alpha\big\}
\end{equation}
for a random variable $\eta$.

\begin{example}
  Assume that $\vecrisk$ has all identical components
  $\risk=\VaR_\alpha$. Furthermore, assume that $C=0$ almost surely on
  a non-atomic probability space, and $\bI(C)=\H$ in $\R^d$, where
  dimension $d$ satisfies $d>\alpha^{-1}$. Partition $\Omega$ into
  subsets $A_1,\dots,A_d$ of probability $d^{-1}$ each. Let
  $\eta^{(n)}_i(\omega)=-(d-1)n$ for $\omega\in A_i$ and
  $\eta_i^{(n)}(\omega)=n$ otherwise, $i=1,\dots,d$. Then $\sum
  \eta^{(n)}_i=0$, so that $\eta^{(n)}\in\H$.  Further,
  $\lim_{n\to\infty} \risk(\eta^{(n)}_i)=-\infty$, since
  $\eta^{(n)}_i=n$ outside a set of probability $d^{-1}<\alpha$.

  Such a construction is not possible if $d<\alpha^{-1}$. In this
  case, the limit property of $\risk(\eta^{(n)}_i)$ yields that for
  any large $a$ and all $i=1,\dots,d$, $\eta^{(n)}_i>a$ outside a set
  $A_i$ of probability at most $\alpha$.  Since $d\alpha<1$, the union
  of all these sets does not cover $\Omega$. This means that all
  components of $\eta^{(n)}$ exceed $a$ simultaneously with positive
  probability, so that $\eta^{(n)}$ is not a selection of $\H$.
\end{example}

\begin{remark}
  In the two-dimensional setting, the existence of the risk arbitrage
  for $\bI(C)=\H$ and the risk measure $\vecrisk=(\risk,\risk)$ means
  that
  \begin{displaymath}
    \risk(\zeta_n)+\risk(-\zeta_n)\to -\infty
  \end{displaymath}
  for a sequence $\zeta_n\in\Lp(\R)$, $n\geq1$. This is clearly
  impossible if $\risk$ is convex or if $\risk=\VaR_\alpha$ with
  $\alpha <1/2$, while it might be the case if $\alpha>1/2$.
\end{remark}

\subsection{Absolute acceptability}
\label{sec:absol-accept-1}

The main setting in the theory of multivariate risk measures concerns
the case of a single agent operating with several currencies or on
various markets. In this case, the existence of an acceptable
selection from the set $\X$ of attainable positions is a natural
acceptability requirement.

In contrary, the interests of agents, and in particular of policy
holders of different legal entities in a group may differ.  When
trying to balance policy holder interests across the whole group, one may
particularly appreciate transfers that do not worsen the situation of
any policy holder, in other words satisfy the \emph{individual
  rationality constraints}, see \cite{jouin:sch:touz08}.

\begin{definition}
  \label{def:abs-risk}
  The pair $(\X(\cdot),C)$ is called \emph{absolutely acceptable} if
  $\X(C)$ admits a selection $\xi\in \Lp(\X(C))$ such that
  $\vecrisk(\xi)\leq 0$ and $\vecrisk(\xi)\leq\vecrisk(C)$.
\end{definition}

In other words, such selection $\xi$ has all individually acceptable
components and each of its components has lower risk than the
corresponding component of $C$. Financially, this may be interpreted
as an admissible IGT that leads to acceptable positions of all agents
without worsening the individual risk assessment of each individual
agent.

The condition $\vecrisk(\xi)\leq \vecrisk(C)$ may be relaxed by
requiring that $\vecrisk(\xi)\in \vecrisk(C)+\K$ for a cone
$\K\subset\R_-^d$ that describes the set of individual risks that are
considered acceptable by all agents within the group.
For a discussion of generalized individual rationality constraints we
refer to~\cite{rav:svi14}.

Clearly, $(\X(C),C)$ is absolutely acceptable if $C$ is acceptable,
while the following example shows that the converse is not
necessarily the case.

\begin{example}
  Let $\X(C)=C+\H$ for $C=(C_1,C_2)$, and let $\vecrisk=(\risk,\risk)$
  have two identical components being the Expected Shortfall
  ($\ES_{0.01}$) at level $0.01$. Assume that $\risk(C_1)<0$, say with
  $C_1$ having the standard normal distribution, and let
  $C_2=\min(a-C_1,0)$ for some $a\geq \E(C_1)$, say $a=\E(C_1)$. Then
  $C_2$ is clearly not acceptable, and so $(C_1,C_2)$ is not
  acceptable. But $(\X(C),C)$ is absolutely acceptable. It is possible
  to reduce the risk of $C_2$ by a transfer without worsening the risk
  of $C_1$, simply because the non-acceptability of $C_2$ stems from
  its behaviour on a set on which $C_1$ takes rather high values. More
  specifically, if $\eta=C_2$, then $(C_1+\eta,C_2-\eta)$ is
  acceptable. Because $C_2-\eta=0$, we have
  $0=\risk(C_2-\eta)<\risk(C_2)$. Furthermore,
  $\risk(C_1+\eta)=\risk(C_1)$, since $\eta=0$ on the event
  $\{\omega:\; C_1(\omega)<a\}$ that has probability at least
  $1/2$. Hence, the $0.01$-quantiles (and all lower quantiles)
  of $C_1$ and $C_1+\eta$ coincide, and we conclude
  $\risk(C_1)=\risk(C_1+\eta)$. Hence $(C_1+\eta,C_2-\eta)$ is
  acceptable and has a componentwisely lower risk than $(C_1,C_2)$.
\end{example}

In many cases, it is impossible to ensure that none of the agents
suffers a deterioration of risk after the optimal risk allocation. The
following result shows that it is possible to achieve individual
rationality after an initial capital transfer.

\begin{proposition}
  \label{prop:prices-crti}
  Assume that there exists $\xi\in\Lp(\X(C))$ such that $\sum
  \risk_i(\xi_i)=\rhot(\X(\cdot),C)$.  Then there
  exists $p=(p_1,\dots,p_d)\in\R^d$ such that $\sum p_i=0$ and
  $\vecrisk(\xi+p)\leq\vecrisk(C)$.
\end{proposition}
\begin{proof}
  The set $\rhox(\X(\cdot),C)$ contains $\vecrisk(C)$. Furthermore,
  \begin{displaymath}
    M=\big(\vecrisk(C)+\R_-^d\big)\cap \Big\{x:\; \sum x_i
    = \rhot(\X(\cdot),C)\Big\}\neq\emptyset,
  \end{displaymath}
  since otherwise $\sum\risk_i(C_i)$ would be lower than the total
  risk. For any $a\in M$, $p=-a+\vecrisk(\xi)$ satisfies the
  requirements.
\end{proof}

The vector $p$ from Proposition~\ref{prop:prices-crti} determines the
prices of risk that the agents pay (or receive) at time zero in order
that the resulting positions do not worsen the risk of any agent and
that the total group risk is the smallest. This result was obtained in
\cite[Th.~3.3]{jouin:sch:touz08} for two agents.

\section{Granular and consolidated tests}
\label{sec:gran-cons-tests}

\subsection{Strictly granular test}
\label{sec:gran-tests}

The \emph{strictly granular} approach presumes that no non-trivial IGTs are
allowed, so that $\bI(C)=\R_-^d$, and
\begin{displaymath}
  \X(C)=C+\R_-^d=(-\infty,C_1]\times\cdots\times (-\infty,C_d]\,.
\end{displaymath}
Then all selections from $\X(C)$ have risks that are not better than
$\vecrisk(C)$, so that $\rhox(\X(\cdot),C)=[\vecrisk(C),\infty)$ and
$\rhot(\X(\cdot),C)=\sum \risk_i(C_i)$ for all monetary risk measures.

\subsection{Consolidated and unconstrained tests}
\label{sec:consolidated-tests}

Recall that the increasingly popular \emph{consolidated} approach
requires the random variable $D$ defined in~(\ref{eq:def-D})
to be acceptable with respect to a prescribed risk measure. It turns
out that, in the coherent case, this setting corresponds to the largest
set of admissible IGTs and would imply 
unrestricted fungibility for all assets,
i.e.\ at the end of the
considered time period assets can be freely used to settle any
liabilities within the group. In this case,
\begin{displaymath}
  \bI(C)=\H=\{x=(x_1,\dots,x_d):\; \sum x_i \leq 0\}
\end{displaymath}
is a half-space, and
\begin{equation}
  \label{eq:xc-half-space}
  \X(C)=\{x=(x_1,\dots,x_d):\; \sum x_i \leq D\}=C+\H\,.
\end{equation}
Note that $\max(0,-D)$ is $p$-integrable if and only if $\X(C)$ has a
$p$-integrable selection.
This may be the case even if $C$ itself is not $p$-integrable.

We call the solvency test that requires $\X(C)$
from~(\ref{eq:xc-half-space}) to be acceptable the \emph{unconstrained
  solvency test}.  In this case,
$\rhox(\X(\cdot),C)=\rhos(\X(C))=\rhos(C+\H)$ and 
\begin{equation}
  \label{eq:total-risk-1}
  \rhot(\X(\cdot),C)
  =\inf\left\{\sum_{i=1}^d \risk(\xi_i):\; \xi=(\xi_1,\dots,\xi_d)\in\
    \Lp(C+\H)\right\}\,.
\end{equation}

In the case of two agents, this situation is
studied in depth in \cite{jouin:sch:touz08}. 
The following result provides an independent analysis of this setting for
$\vecrisk$ with all identical components and extends it for the convex
case and any number of agents.

\begin{theorem}
  \label{thr:trans-consolidated}
  Let $C=(C_1,\dots,C_d)$ be a $p$-integrable random vector and let
  $\vecrisk=(\risk,\dots,\risk)$ for a monetary $\Lp$-risk measure
  $\risk$. Furthermore, let $\X(C)$ be given by
  \eqref{eq:xc-half-space}.
  \begin{enumerate}[i)]
  \item If $\risk$ is coherent, then $\X(C)$ admits an acceptable
    selection if and only if $\risk(D)\leq 0$. In this case, the
    infimum in \eqref{eq:total-risk-1} is attained at
    $\xi=d^{-1}(D,\dots,D)$ and $\rhot(\X(\cdot),C)=\risk(D)$.
  \item Assume that, for any $p$-integrable random vector
    $C$, $\X(C)$ admits an acceptable selection if and only if
    $\risk(D)\leq 0$. Then $\risk$ is subadditive.
  \item If $\risk$ is convex, then $\X(C)$ admits an acceptable selection
    if and only if $\risk(\frac{D}{d})\leq 0$.  The
    infimum in \eqref{eq:total-risk-1} is attained at
    $\xi=d^{-1}(D,\dots,D)$ and $\rhot(\X(\cdot),C)=d\risk(D/d)$. 
  \item Assume that, for any $d$ and for any $(C_1,\dots,C_d)$, the
    acceptability of the random vector $(C_1,\dots,C_d)$ yields the
    acceptability of $d^{-1}(D,\dots,D)$. Then $r$ is convex.
  \end{enumerate}
\end{theorem}
\begin{proof}
  \textsl{i)} Assume that $\risk(D)\leq 0$ and define
  $\xi=d^{-1}(D,\dots,D)$, which is a selection of $\X(C)$. Since
  $\risk$ is positive homogeneous, all components of $\xi$ are
  acceptable, and $\eta=\xi-C$ yields the corresponding IGT. Note
  that the sum of coordinates of $\eta$ vanishes almost surely.

  Conversely, assume that there exists a selection $\xi$ of $\X(C)$
  such that $\vecrisk(\xi)\leq 0$. Let $\hat\xi$ be the projection of
  $\xi$ onto the boundary of $\X(C)$. Note that $\sum \hat\xi_i \geq
  \sum\xi_i$. Then $\hat \xi=C+\eta$, with $\eta=\hat\xi-C$ such that
  $\sum\eta_i=0$ a.s.  Hence,
  \begin{align*}
    \risk(D)=\risk\Big(D+\sum \eta_i\Big) =\risk\Big(\sum \hat\xi_i\Big)
    \leq \risk(\sum \xi_i)\leq
    \sum \risk(\xi_i)\leq 0\,.
  \end{align*}
  The infimum in \eqref{eq:total-risk} equals the infimum of all
  $a=\sum x_i$ such that $\risk(D+a)\leq 0$ and so is attained.

  \textsl{ii)} For any $p$-integrable random vector $C$, define
  $C'=C+\vecrisk(C)$.  By the monetary property, $\vecrisk(C')=0$,
  so every $C'_1, \dots, C'_d$ is acceptable. By the assumption,
  $D'=\sum C'_i$ is acceptable, while the monetary
  property yields that
  \begin{displaymath}
    0\geq\risk(D')=\risk\Big(\sum C_i +\sum \risk(C_i)\Big)
    =\risk\Big(\sum C_i\Big)-\sum \risk(C_i)
  \end{displaymath}
  as desired.

  \textsl{iii)} If $\risk(D/d)\leq 0$, then
  $\xi=d^{-1}(D,\dots,D)$ is an acceptable selection. Conversely,
  assume $\xi=(\xi_1,\dots,\xi_d)$ is an acceptable selection.  Then
  $\sum \xi_i\leq \sum C_i$, hence $\frac1d \sum \xi_i \leq
  \frac1d \sum C_i$. By convexity, 
  \begin{displaymath}
    \risk(D/d)\leq
    \risk\Big(\frac{1}{d}\sum \xi_i\Big)\leq \frac{1}{d} \sum \risk(\xi_i)\leq 0.
  \end{displaymath}
  
  \textsl{iv)} By the cash invariance, the assumption is equivalent to
  \begin{displaymath}
    r(D/d)\leq \frac{1}{d} \sum r(C_i)\,.
  \end{displaymath}
  We have to show that
  \begin{displaymath}
    r(\lambda C_1 + (1-\lambda) C_2)
    \leq \lambda r(C_1)+(1-\lambda) r(C_2)
  \end{displaymath}
  for any $C_1, C_2$ and $0\leq \lambda \leq 1$.  Due to the strong
  continuity of the components of $\vecrisk$, it suffices to show this
  for rational $\lambda=\frac{m}{n}$.  Applying the assumption to the
  random vector $(C_1,\dots,C_1,C_2,\dots,C_2)$ consisting of $m$
  copies of $C_1$ and $n-m$ copies of $C_2$ yields that
  \begin{displaymath}
    r\left(\frac{mC_1+(n-m)C_2}{n}\right)
    \leq \frac{m}{n}r(C_1)+\frac{n-m}{n}r(C_2),
  \end{displaymath}
  hence, the assertion is proved.
\end{proof}

If the components of $\vecrisk$ are not necessarily coherent, then the
classical consolidated approach and the unconstrained approach do not
necessarily result in the same risks.  If $D=\sum C_i$ is acceptable,
then $C+\H$ admits an acceptable selection, which is given, e.g., by
$(D,0,\dots,0)$. However, neither the acceptability of $D$ nor of
$\frac1d D$ can be concluded from the existence of an acceptable
selection in the non-convex case, see Proposition~\ref{prop:non-convex-D}.  Thus, the group
risk calculated on the basis of the unconstrained solvency test may be
too optimistic in comparison with the classical consolidated test in
the non-convex setting.

\begin{proposition}
  \label{prop:non-convex-D}
  Assume that $d\geq2$ and $\vecrisk$ has identical components being
  the Value-at-Risk.
  \begin{itemize}
  \item[(i)] If $C_1,\dots,C_d$ are acceptable for $\VaR_\alpha$, then
    $D=\sum C_i$ is acceptable for $\VaR_{d\alpha}$.
  \item[(ii)] 
    If $\beta<d\alpha$, then there are random variables
    $C_1,\dots,C_d$ such that $C_i$ is acceptable for $\VaR_\alpha$,
    but $D$ is not acceptable for $\VaR_\beta$. In particular, this is
    the case for $\beta=\alpha$. 
  \end{itemize}
\end{proposition}
\begin{proof}
    (i) The assumption yields $\P(C_i<-\eps)<\alpha$, for each
  $\eps>0$. Then
  \begin{displaymath}
    \P(D<-d\eps)\leq \sum \P(C_i<-\eps)<
  d\alpha\,,
  \end{displaymath}
  implying the acceptability of $D$ under $\VaR_{d\alpha}$.

  (ii) Choose a non-atomic probability space and let
  $\Omega_1,\dots,\Omega_d$ be mutually disjoint events of probability
  $\tilde\alpha=\alpha-(d\alpha-\beta)/(2d)$. Define
  $C_i(\omega)=-1$ for $\omega\in\Omega_i$, and $C_i(\omega)=0$
  otherwise. Then $\VaR_\alpha(C_i)\leq 0$ for all $i$, while
  \begin{displaymath}
    \P\left(D<-\frac{1}{2}\right)
    =\P(\cup_{i=1}^{d}\Omega_i)
    =\sum_{i=1}^{d}\P(\Omega_i)=d\tilde\alpha
    =\frac{1}{2}(d\alpha+\beta)>\beta,
  \end{displaymath}
  so that $D$ is not acceptable for $\VaR_\beta$.
\end{proof}

\section{Restrictions of transfers}
\label{sec:restr-transf}

\subsection{Restrictions of transfers and the total risk}
\label{sec:restr-transf-and-tot-risk}

In the general single currency setting, the sets of admissible
terminal portfolios are sandwiched between the strictly granular
approach and unconstrained approach. Which IGTs are accepted as
admissible is not primarily a mathematical question.
In the sequel we show how considerations regarding solvency
assessment, transaction costs, and some sources of risks might be
reflected in the design of the set $\bI(C)$ of admissible IGTs.

The choice $\bI=\H$ corresponds to the unconstrained risk
assessment and so yields the largest possible set of admissible IGTs
in the single currency setting.

The monotonicity property from
Proposition~\ref{prop:tot-risk-mon} in a coherent case yields
that, under any restrictions, the group will never achieve better
risks comparing to those obtained by performing the unconstrained
solvency test approach. This is formalized in the following
proposition.

\begin{proposition}
  Let $C=(C_1,\dots,C_d)$ be a $p$-integrable random vector and let
  $\vecrisk=(\risk,\dots,\risk)$ for a convex $\Lp$-risk measure
  $\risk$. If $\X(C)$ is any family of admissible IGTs, then the
  corresponding total risk is at least $d\,\risk(d^{-1}\sum C_i)$, which
  is the total risk of the unconstrained solvency test.
\end{proposition}

A restriction of $\H$ increases the complexity of the set-valued
solvency tests considerably. In view of that it is natural to ask
whether the impact on the resulting capital requirements is
material or not compared to the unconstrained and strictly
granular approaches.

\begin{proposition}
  \label{prop:no-change-tot-risk}
  Assume that $\vecrisk=(\risk,\dots,\risk)$ for a coherent $\Lp$-risk
  measure $\risk$.  If the unconstrained total group risk does not
  increase by restricting the transfers to $\bI(C)$, then, for some
  $x=(x_1,\dots,x_d)$ with $\sum x_i = \risk(\sum C_i)$, the set
  $\bI(C+x)$ contains a $p$-integrable selection $\eta$ such that
  \begin{equation}
    \label{minrisk}
    \risk(\sum C_i)= \sum \risk(C_i+\eta_i).
  \end{equation}
\end{proposition}
\begin{proof}
  If the total risk in the restricted setting does not increase in
  comparison with the unrestricted one, then there exists $x$ with
  $\sum x_i\leq \risk(\sum C_i)$, such that
  \begin{displaymath}
    \sum \risk(C_i+\eta_i+x_i)\leq 0
  \end{displaymath}
  for a selection $\eta$ of $\bI(C+x)$. The monetary property yields
  that
  \begin{displaymath}
    \sum \risk(C_i+\eta_i)\leq \sum x_i\leq \risk(\sum C_i),
  \end{displaymath}
  and the subadditivity property together with $\sum \eta_i\leq 0$
  yield the equality.
\end{proof}

\subsection{No transfers causing or worsening bankruptcy}
\label{sec:NTB}

For any policy holder, one of the most important events to be
avoided 
is bankruptcy of their counter-party.  Hence, in order to balance the
interests of all policy holders of all legal entities within the
group, a natural restriction for admissible IGTs could be to exclude
transfers that exceed the capital of a legal entity, if this capital
is positive. Furthermore, it could also be argued that it would not be
in line with policy holder interests of the bankrupt company if
their bankruptcy's dividend were reduced by intragroup
transactions.  It is also clear that fungibility is dramatically
restricted in case of bankruptcy, e.g.\ to a certain bankruptcy's
dividend if legally binding and enforceable intragroup contract
exist.
Hence, in order to at least partially protect policy holder interests
of a bankrupt subsidiary and in order to include some fungibility
aspects it seems reasonable to prohibit
transfers out of a company with negative capital.

To sum up, there are several reasons to exclude transfers turning a
non-bankrupt legal entity into a bankrupt one and also transfers out
of a bankrupt one to another one (from the same group). We use the
abbreviation \NTB\ (\emph{No Transfers causing or worsening
  Bankruptcy}) for this kind of IGTs.  Here we make the simplifying
assumption that bankruptcy is defined with respect to the terminal
capital position $C$ and not with respect to a different balance
sheet, i.e.\ transfers may not turn a non-negative component of $C$
into a negative one.

The corresponding set of attainable positions $\X(C)=C+\bINTB(C)$ is
given by
\begin{align*}
  \bINTB(C)&=\{(x_1,\dots,x_d):\;\sum x_i\leq 0,x_i\geq-C_i^+,\; i=1,\dots,d\},
\end{align*}
where $a^+=\max(a,0)$ for $a\in\R$. Note that \eqref{eq:iteration}
holds in this case and 
$\X(C+x)$ non-linearly depends on $x$. Furthermore, $\partial^+\X(C)$ is
$p$-integrably bounded and upper semicontinuous, so that
$\rhos(\X(C))$ and the group risk $\rhox(\X(\cdot),C)$ are closed
sets in the coherent case, see Proposition~\ref{prop:closed-total}.  

For a group consisting of two agents, the set $\X(C)$ of terminal
positions has vertices at $(C_1+C_2^+,C_2-C_2^+)$ and
$(C_1-C_1^+,C_2+C_1^+)$, see Figure~\ref{fig:ntb}. Therefore,
\begin{equation}
  \label{eq:support-ntb}
  h_{\X(C)}(u)=\langle C,u\rangle+
  \begin{cases}
    C_2^+(u_1-u_2),  & u_1\geq u_2,\\
    C_1^+(u_2-u_1),  & u_1<u_2,
  \end{cases}
\end{equation}
for $u=(u_1,u_2)\in\R_+^2$. 

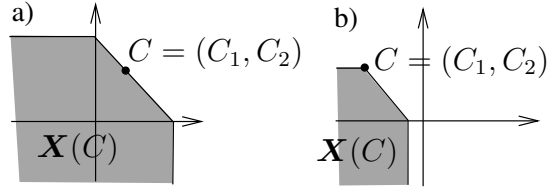
\begin{figure}
  \centering
  \input{ntb.pspdftex}
  \caption{The set of attainable positions in case of both agents are solvent (a) and in case the second agent is solvent (b). }
  \label{fig:ntb}
\end{figure}

\subsection{Safety margin}
\label{sec:other-examples}

Allowing transfers to vanishing capital for some agents (as it is
possible under \NTB) may still be considered too progressive,
since the agents might end up with no capital buffer after IGTs, i.e.\
they are almost bankrupt.  In view of that it is worth noticing that
requirements on admissible IGTs can be made more stringent if the set
$\bI(C)$ is replaced by $\bI(C-a)$ for a fixed vector
$a=(a_1,\dots,a_d)$ with non-negative components that set
\emph{safety margins} for terminal capital amounts. If $\X(C)=C+\bI(C-a)$, then
\begin{align*}
  \rhox(\X(\cdot),C)&=\{x:\; \rhos(C+x+\bI(C+x-a))\ni 0\}\\
  &=a+\{x:\; \rhos(\X(C+x))\ni a\},
\end{align*}
so that the group risk is obtained by inverting the selection risk
measure at point $a$, cf.~\eqref{eq:rhox}.

In the case of proportional safety margins, $\bI(C)$ is replaced by
$\bI(C-\lambda C^+)$, where $\lambda
C^+=(\lambda_1C_1^+,\dots,\lambda_dC_d^+)$ for $\lambda\in [0,1]^d$.

\subsection{Fungibility costs}
\label{sec:add-fungibility}

So far we have assumed that capital can be transferred either in full
or not at all. Due to fungibility constraints, the capital may only
flow via complicated constructs that involve taking loans and cause
serious future fungibility costs, in particular for transitional
funding.  This can be included in our framework by adapting
the set $\X(C)$ of attainable terminal positions.

In the interest of brevity, we illustrate the construction with the
help of a bivariate example.  In the single currency setting
$\partial^+\X(C)$ is a curve that passes through $C=(C_1,C_2)$ and has
slope $-1$ at that point if $C_1,C_2>0$. The fungibility difficulties may be
modeled by changing the slope of this line according to increasing
fungibility costs for agents with low capital.  If the capital amounts
of the firms are large, and the transfers are small (so that the
capital amounts after transfer exceed thresholds $\bar c_1$ and $\bar
c_2$), it is well possible that fungibility costs from one company to
another one vanish.  However, if the capital of a company is small and
the transfer is relatively large, then this can cause serious
fungibility costs, which are getting larger the closer the company
comes to bankruptcy after a transfer. Then e.g.\ for a modified \NTB\
setting the boundary $\{x=(x_1,x_2):\; x\in\partial^+\X(C)\}$ can
be modeled using two differential equations
\begin{equation}
  \label{eq:ode-fungi}
  \frac{\partial x_2}{\partial x_1}
  =-\left(\frac{\bar c_2}{x_2}\right)^p\,,\qquad
  \frac{\partial x_1}{\partial x_2}
  =-\left(\frac{\bar c_1}{x_1}\right)^p
\end{equation}
depending on whether we want to transfer money from the second to the
first company (first equation, $0<x_2<\bar c_2$, in the considered
modified \NTB\ setting only relevant if $C_2>0$), or from the first to
the second company (second equation, $0<x_1<\bar c_1$, in the
considered modified \NTB\ setting only relevant if $C_1>0$). Close to
bankruptcy, the fungibility costs become immense and for all
non-positive values of capital of the donor company no transfer is
possible anymore. The parameter $p\geq 1$ can be used for calibration
to the company specific situations. The solution of these two
equations together with the initial condition $(x_1,x_2)=(C_1,C_2)$
produces a curve on the plane that yields $\partial^+\X(C)$. Note that
$\X(C) =\partial^+\X(C)+\R_-^2$ satisfies \eqref{eq:iteration}.

\subsection{Transaction costs}
\label{sec:transaction-costs}

Consider the case, where transfers are subject to transaction costs,
however unrestricted otherwise. In case of proportional transaction
costs, each recipient surrenders a proportion of the amount
determined by a factor $\pi\in[0,1]$. For simplicity, we assume $\pi$
to be deterministic and the same for all legal entities. For two legal
entities, the set of admissible IGTs for proportional transaction
costs is
\begin{displaymath}
  \bI=\{x=(x_1,x_2):\; x_1+\pi x_2\leq 0,\pi x_1+x_2\leq 0\},
\end{displaymath}
which is a convex cone that does not depend on the capital position
$C$. The corresponding risks have been studied in depth in
\cite{ham:hey10,ham:hey:rud11} and \cite{cas:mol14}. The results of
these papers for deterministic exchange cones apply in our
setting. In particular, $\rhox(\X(\cdot),C)=\rhos(C+\bI)$ is the selection
risk measure of the set-valued portfolio $C+\bI$.

In case of fixed transaction costs, each recipient surrenders a fixed
amount $a\geq0$, so that
\begin{displaymath}
  \bI=\R_-^2\cup \{x=(x_1,x_2):\; x_1+x_2\leq -a\}.
\end{displaymath}
This provides an example of a non-convex set of IGTs that also does
not depend on $C$.

A particularly important case of transaction costs relates to the case
where agents operate in different currencies. Then transfers between
the currencies are subject to transaction costs and may also involve
random exchange rates. In this case, it is also problematic to consider
the total risk, since there is no natural reference currency to
express the risks of all agents. We leave this setting for future
work.

\section{Unequally placed agents}
\label{sec:uneq-plac-agents}

Most financial groups exhibit some hierarchical structures, namely
there are parent and subsidiary agents, and even cross-holdings
between different members of the same group often exist. For
simplicity, consider the case of two agents: the first agent with
capital $S$ is a subsidiary and the second one with capital $C_2$ is a
holding (or parent).  The parent company owns an option on the (full)
available capital of the subsidiary, e.g.\ via liquidation, i.e.\ the
parent has a long position in the derivative with payoff
$S^+=\max(0,S)$. Assume that the parent has no other assets or
liabilities, so that her capital is $S^+$.

In this case, adding capital amounts would incur double counting of the
positive part of $S$. Such double counting is typically excluded in
previous works on optimal risk sharing, see~\cite{fil:kun08}.  In our
framework, it is possible to avoid double counting by adjusting
the random set of attainable positions.

In view of the fact that a parent company ultimately has the legal
right to get the (full) capital from the subsidiary, without an
additional contract the parent can add the positive part of the
capital of the subsidiary to its assets, and hence, to its
capital. Consequently the subsidiary acquires the corresponding short
position, which has to be considered in the risk calculation of the
subsidiary, so that its consistent capital becomes $-S^-\leq 0$, where
$S^-=(-S)^+$ is the so-called limited liability put option. Thus, the
net capital of the subsidiary is never strictly positive, for
instance, the subsidiary would never be acceptable, except in cases
where the capital of the subsidiary after taking into account the
participations is concentrated at zero, which might happen, if the
subsidiary were ``long only'' in assets while being completely
financed by equity.

For the capital position $C=(C_1,C_2)=(-S^-,S^+)$, in the strictly
granular setting the random set of attainable positions is given by
\begin{displaymath}
  \Y(C)=
  \begin{cases}
    (-\infty,0]\times (-\infty,S] & \text{if} \; S\geq 0,\\
    (-\infty,S]\times (-\infty,0] & \text{otherwise},
  \end{cases}
\end{displaymath}
so that $\partial^+\Y(C)=(-S^-,S^+)$.  Assume that the risks of the
both group members are assessed using the same coherent risk measure
$\risk$.  Thus, in the strictly granular setting, the total risk
amounts to
\begin{displaymath}
  \rhot(\Y(\cdot),C)=\risk(S^+)+\risk(-S^-)\geq \risk(S^+-S^-)=\risk(S).
\end{displaymath}

Allowing for IGTs, for any state of the world one has to fix the
transfer from the parent to the subsidiary, which is the only feasible
transfer direction. For this, the parent obtains a loan $a$ secured
upon $S^+$, and we do not take into account any fungibility
difficulties that might occur in this relation. From this loan a
random amount $\eta$ is due to be transferred to the subsidiary at the
terminal time.  If $S\geq0$, then the transferred amount is
immediately recovered and the loan is repaid, while if $S<0$, then the
parent recovers $(S+\eta)^+$. The set $\X(C)$ of attainable positions
is characterized by $\partial^+\X(C)$ given by
$(-(S+\eta)^-,(S+\eta)^+-\eta)$ for all possible transfers $\eta$,
which are random variables taking values in $[0,a]$.
The corresponding ``total risk'' becomes
\begin{displaymath}
  \risk(-(S+\eta)^-)+\risk((S+\eta)^+-\eta)\geq \risk(S)
\end{displaymath}
in view of the subadditivity of the risk measure. In the hypothetical
case of an unlimited credit line, the equality, and thus, the optimal
transfer is achieved if $\eta=S^-$.

The policy holders of a
parent could benefit from the investment in the subsidiary without
directly affecting reserves of the policy holders of the subsidiary,
if the parent sells the subsidiary at the terminal time to a third
party.  Since this strategy needs a buyer in a potentially stressed
market it is possibly, from the initial time point of view,
recommendable not to fully take into account this possibility in a prudent
regulation.

\section{Diversification effects}
\label{sec:divers-effects}

Filipovi{\'c}~and Kunz~\cite{fil:kun08} present a 
bottom-up approach to analyse intragroup diversification in a very
concrete setting with given distributions and choice of predefined
IGTs. For a rather recent similar analysis, see also~\cite{mas13}.

The key property of coherent risk measures is their subadditivity that
corresponds to the fact that diversification decreases risk. The
non-linear feature of the IGTs brings new features to the
diversification effects. For example, in many cases, e.g.\ in the \NTB\
case,
\begin{equation}
  \label{eq:reverse-effects}
  \bI(C'+C'')\subset \bI(C')+\bI(C'')\,,
\end{equation}
where both $C'$ and $C''$ are $d$-dimensional vectors of capital
values. Then the diversification of assets and liabilities narrows the
range of admissible IGTs. On the other side, there exists a classical
benefit from diversification effects. A similar situation arises in
the diversification effects for systemic risks, see
\cite{fein:rud:web15}.

\begin{example}[Univariate case]
  In order to understand the diversification effects in the group
  setting, consider the one-dimensional case from the point of view of
  the group. Let $C'$ and $C''$ be two $p$-integrable random variables
  and let $\risk$ be a univariate coherent $\Lp$-risk measure. Define
  $C=(C',C'')$ to be a random vector in $\R^2$ and let
  $\vecrisk=(\risk,\risk)$. In view of the subadditivity of $\risk$
  the total risk of $C$ with the IGTs given by $\bI(C)=\H$ equals
  $\risk(C'+C'')$, while if $\bI(C)=\R_-^2$, then the total risk
  becomes $\risk(C')+\risk(C'')$. Thus, the classical diversification
  benefit can be phrased as the advantage that corresponds to
  increasing the set of admissible IGTs from $\R_-^2$ to $\H$, in
  other words, from altering the strictly granular approach to the
  unconstrained one. Related to that, it should be stressed that a
  solvency test should never include unrealistic transfers, since
  otherwise the solvency test tends to underestimate the real risks a
  group is faced with. Furthermore, it should be noted that classical
  diversification compares $\risk(C')+\risk(C'')$ with the risk of a
  particular selection of $\X(C)=C+\H$, namely that of $(C'+C'',0)$.
\end{example}

The classical concept of diversification is inherent for a single
agent, who might have several business units with unrestricted capital
flows between them. In case of groups, we see two basic effects:
\begin{itemize}
\item \emph{consolidation} that amounts to increasing the set of
  admissible IGTs;
\item \emph{granularization} that corresponds to restricting the
  family of admissible IGTs.
\end{itemize}
In particular, a merger of two legal entities removes all fungibility
barriers between them, as in the case of unconstrained approach, and
so is a simple example of consolidation.  On the contrary, a split may
lead to some
additional restrictions in capital transfers that can be viewed as
granularization.

\begin{example}
  Consider the group $C=(C_1,\dots,C_d)$ and assume that the first
  agent splits its operation into two subsidiaries so that
  $C_1=C_{11}+C_{12}$. The effect of such granularization on the total
  risk depends on the set of admissible transfers between the two
  created subsidiaries and between them and the rest of the group. For
  instance, the total risk is retained if the two subsidiaries are
  considered on the unconstrained basis.
\end{example}

\begin{example}
  Consider a single agent with terminal capital $\tilde C$ that is
  split into $C=(C_1,C_2)$, so that $\tilde C=C_1+C_2$. The total risk
  after such granularization under a coherent risk measure $\risk$
  lies between $\risk(C_1+C_2)=\risk(\tilde C)$ in the unconstrained
  case and $\risk(C_1)+\risk(C_2)$ in the strictly granular
  setting. 
\end{example}

\begin{example}
  If two groups are merged, a classical question is, whether or not
  also some of the legal entities (like two life- or non-life
  companies) should also be merged.  Consider two groups
  $C'=(C'_1,\dots,C'_{d'})$ and $C''=(C''_1,\dots,C''_{d''})$. Their
  merge creates a new group $C$ with $d'+d''$ legal entities, so that
  one has to specify the family $\bI(C)\subset\R^{d'+d''}$ of
  admissible IGTs. It is natural to assume that $\bI(C)\supset
  \bI'(C')\times\{0\}$ and $\bI(C)\supset \{0\}\times \bI''(C'')$,
  meaning that the families $\bI'(C')$ and $\bI''(C'')$ of admissible
  IGTs within each of two primary groups are admissible after the
  merge. Still the larger group may be allowed to transfer capital
  between the components of $C'$ and $C''$. In view of this,
  \begin{displaymath}
    \rhot(\X(\cdot),C)\leq \rhot(\X'(\cdot),C')+\rhot(\X''(\cdot),C'').
  \end{displaymath}
  Note that the coherence of components of $\vecrisk$ is not needed
  for this.  Hence, if two groups are merged without any merger of the
  legal entities we can expect (e.g.\ if the above assumptions are
  satisfied) the total risk to be subadditive. However, if we start
  to merge also legal entities, the situation is less clear due to
  reverse effects from~\eqref{eq:reverse-effects}.
\end{example}

Thus, in the context of risk assessment for groups, the diversification
advantage can be formulated as follows.
\begin{quote}
  Granularization does not diminish the risk, while
  consolidation does not increase the risk.
\end{quote}
This fact is the sole monotonicity property of the set of admissible
IGTs and does not rely on the subadditivity property of the involved risk
measures.

\section{Calculating the group risk}
\label{sec:num-examples}

The calculation of the group risk requires finding $x$ such that the
selection risk measure of $\X(C+x)$ contains the origin. This is a
serious computational problem, that can be solved by means of
multicriterial optimization algorithms, see
e.g.~\cite{ham:rud:yan13}. However, bounds on the group risk can be
obtained as follows, see Theorem~\ref{cor:lower-bound}.

\begin{proposition}
  \label{prop:bounds}
  Let $\vecrisk=(\risk,\dots,\risk)$ be a coherent $\Lp$-risk measure
  with identical components. Let $\xi(x)$ be any selection of
  $\X(C+x)$ for $x\in\R^d$. Then
  \begin{equation}
    \label{eq:bounds}
    \{x\in\R^d:\; \vecrisk(\xi(x))\leq 0\}\subset \rhox(\X(\cdot),C)
    \subset\bigcap_{u\in\R_+^d} \{x\in\R^d:\;
    \risk(h_{\X(C+x)}(u))\leq 0\}.
  \end{equation}
\end{proposition}

In the following, assume that $\vecrisk$ has all identical coherent
components.  In case of two agents, the calculation of
the superset for the group risk can often be simplified by the following
proposition.

\begin{proposition}
  \label{prop:convex-combination}
  Assume that $\X(C)$ is almost surely convex in $\R^2$ with
  $\partial^+\bI(C)\subset \{x=(x_1,x_2):\; x_1+x_2=0\}$. Then the
  superset in \eqref{eq:bounds} does not change if the intersection is
  reduced to $u\in\{(1,0),(0,1),(1,1)\}$.
\end{proposition}
\begin{proof}
  It suffices to show that if $\risk(h_{\X(C+x)}(u))\leq 0$ for the
  three above mentioned $u$, then the inequality holds for all
  $u\in\R_+^2$. Without loss of generality assume that $x=0$.  The
  condition means that $\partial^+\bI(C)$ is the segment with two
  end-points $(\zeta_1,-\zeta_1)$ and $(-\zeta_2,\zeta_2)$, where
  $\zeta_1,\zeta_2$ are two non-negative random variables that might
  depend on $C$. Then
  \begin{displaymath}
    h_{\X(C)}(u)=\langle C,u\rangle+
    \begin{cases}
      \zeta_1(u_1-u_2) & \text{if } u_1>u_2,\\
      \zeta_2(u_2-u_1) & \text{if } u_2\geq u_1.
    \end{cases}
  \end{displaymath}
  If $u_1>u_2$, then
  \begin{displaymath}
    h_{\X(C)}(u)=(u_1-u_2)h_{\X(C)}((1,0))+u_2h_{\X(C)}((1,1)).
  \end{displaymath}
  The coherency of the risk measure yields that the risk of
  $h_{\X(C)}(u)$ is acceptable if both $h_{\X(C)}((1,0))$ and
  $h_{\X(C)}((1,1))$ are. The case of $u_2>u_1$ is similar.
\end{proof}

In the following we choose the \NTB\ restrictions for two agents,
possibly with safety margins, where
Proposition~\ref{prop:convex-combination} clearly applies. In case of
a fixed safety margin $(a_1,a_2)$, the group risk is a subset of
\begin{multline}
  \label{eq:three-planes}
  \{x=(x_1,x_2):\; \risk(C_1+x_1+(C_2+x_2-a_2)^+)\leq 0,\\
  \risk(C_2+x_2+(C_1+x_1-a_1)^+)\leq 0, \risk(C_1+C_2)\leq x_1+x_2\}.
\end{multline}
Note that the last inequality defines a half-plane corresponding to
the unconstrained setting.  In the case of zero safety margin, the
first two inequalities are superfluous, and so the superset from
Proposition~\ref{prop:bounds} does not differ from the group risk in
the unconstrained setting. Indeed, if $x_1+x_2\geq \risk(C_1+C_2)$,
then without loss of generality we can assume the
equality, so that $x_1=\risk(C_1+C_2)-x_2$. Then the first inequality
in \eqref{eq:three-planes} requires
\begin{displaymath}
  \risk(C_1+(C_2+x_2)^+)\leq \risk(C_1+C_2)-x_2,
\end{displaymath}
or, equivalently,
\begin{displaymath}
  \risk(C_1+C_2+(C_2+x_2)^+-(C_2+x_2))\leq \risk(C_1+C_2),
\end{displaymath}
which always holds, since $(C_2+x_2)^+\geq(C_2+x_2)$.

In order to obtain a subset of $\rhox(\X(\cdot),C)$, choose the selection
$\xi(x)$ as the point of $\partial^+\X(C+x)$ that is nearest to the
diagonal line $\{(x_1,x_2):x_1=x_2\}$. Then $\xi(0)=\frac{1}{2}(D,D)$
if $D=C_1+C_2\geq 0$. If $D<0$, then $\xi(0)=(D,0)$ in case
$C_2\geq 0$, $\xi(0)=(0,D)$ if $C_1\geq0$, and $\xi(0)=(C_1,C_2)$ if
$C_1,C_2<0$. Thus, $\xi(0)$ is acceptable if
\begin{align*}
  \risk(\frac{1}{2} D
  \one_{D\geq0}+D\one_{D<0,C_2\geq 0}+C_1\one_{C_1<0,C_2<0})&\leq 0,\\
  \risk(\frac{1}{2} D
  \one_{D\geq0}+D\one_{D<0,C_1\geq 0}+C_2\one_{C_1<0,C_2<0})&\leq 0.
\end{align*}
The inner approximation in \eqref{eq:bounds} is obtained as the set of
all $x=(x_1,x_2)$ such that the above inequalities hold with
$C=(C_1,C_2)$ replaced by $(C_1+x_1-a_1,C_2+x_2-a_2)$, where
$(a_1,a_2)$ denotes the fixed safety margin. An upper bound for the
total risk is obtained by finding the minimum of $x_1+x_2$ that
satisfy these inequalities.

\begin{example}
  Consider the \NTB\ setting without safety margin and with
  $\vecrisk=(\ES_{0.01},\ES_{0.01})$ as the underlying risk measure.
  Figure~\ref{fig:uni-exp}(a) shows the superset and the subset of the
  group risk for i.i.d.\ $C_1,C_2$ having the uniform distribution on
  $[0,5]$. The outer bound equals the group risk in
  the unconstrained setting. The upper right corner shows the group
  risk in the strictly granular case. While the bounds for
  $\rhox(\X(\cdot),C)$ given by Proposition~\ref{prop:bounds} clearly
  differ, they both yield the same value of the total risk. This means
  that the total risk in the \NTB\ setting for this example
  coincides with the unconstrained group risk.
  While the same phenomenon appears in all simulated examples with
  exchangeable $(C_1,C_2)$, we do not have a theoretical confirmation
  of this observation.

  Figure~\ref{fig:uni-exp}(b) shows the bounds for the group risk in
  case $C_1$ has the standard normal distribution independent of
  $-C_2$ having the exponential distribution of mean one. In this case
  the inner approximation to the group risk does not touch the outer
  approximation, however close they are.

  \begin{figure}
    \centering
    \begin{tabular}{cc}
      \subfigure[]{\includegraphics[width=6cm,height=6cm]{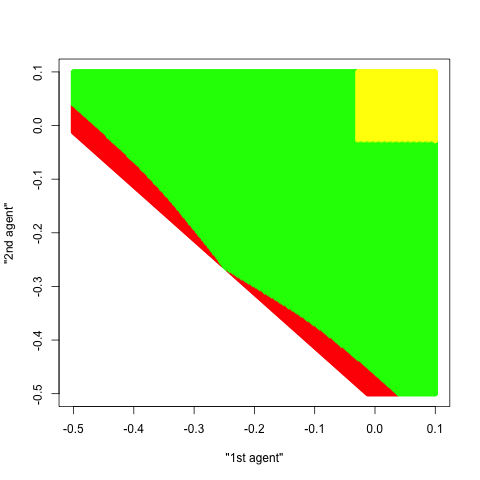}}
      &
      \subfigure[]{\includegraphics[width=6cm,height=6cm]{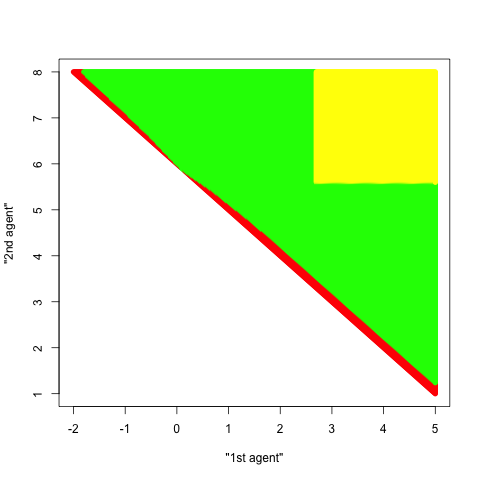}}
    \end{tabular}
    \caption{Bounds for the group risk for the uniform distribution
      (a) and for combination of normal and exponential distributions
      (b). The strictly granular group risk is shown in yellow, the
      inner approximation in green, the outer approximation in red (it
      coincides with the unconstrained group risk).}
    \label{fig:uni-exp}
  \end{figure}
\end{example}

\begin{example}
  Take $\vecrisk=(\ES_{0.01},\ES_{0.01})$ and
  let $(C_1,C_2)$ be normally distributed with mean zero and the
  covariance matrix
  \begin{displaymath}
    \begin{pmatrix}
      1 & -0.5 \\
      -0.5 & 3
    \end{pmatrix}.
  \end{displaymath}
  Figure~\ref{fig:norm}(a) shows approximations to the group risk
  without safety margin. Here the outer approximation coincides with
  the group risk in the unconstrained setting, while the inner
  approximation touches it and so shows that the total risk in the
  \NTB\ setting coincides with the unconstrained total risk.
  Figure~\ref{fig:norm}(b) shows the results for the fixed safety
  margin set to $0.5$ for the both agents. In this case the outer
  approximation in \eqref{eq:bounds} coincides with the inner
  approximation and so yields the group risk. The outer set
  corresponds to the unconstrained
  setting. 
  An indication for the high potential of IGTs for intragroup
  diversification is seen by comparing the strictly granular
  group risk (shown in yellow) with other ones.

  \begin{figure}
    \centering
    \begin{tabular}{cc}
      \subfigure[]{\includegraphics[width=6cm,height=6cm]{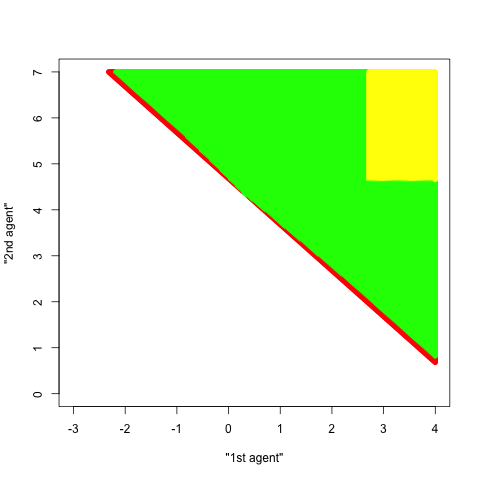}}
      &
      \subfigure[]{\includegraphics[width=6cm,height=6cm]{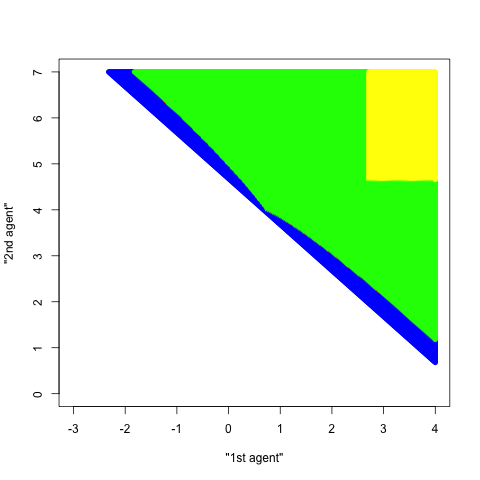}}
    \end{tabular}
    \caption{Bounds for the group risk for the normally distributed
      assets without safety margin (a) and with fixed safety margin of
      $0.5$ for both agents (b). The strictly granular group risk
      is shown in yellow, the inner approximation in green, the outer
      approximation in red, and the unconstrained group risk in blue.}
    \label{fig:norm}
  \end{figure}
 \end{example}




\end{document}

%% file: ntb.pspdftex
\begin{picture}(0,0)%
\includegraphics{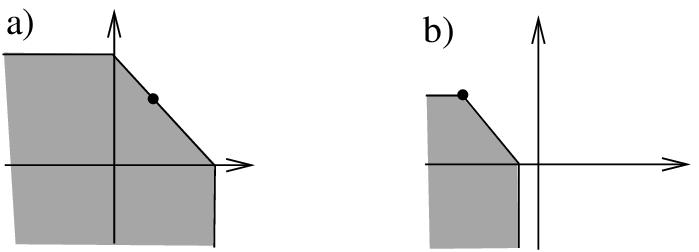}%
\end{picture}%
\setlength{\unitlength}{3947sp}%
\begingroup\makeatletter\ifx\SetFigFont\undefined%
\gdef\SetFigFont#1#2#3#4#5{%
  \reset@font\fontsize{#1}{#2pt}%
  \fontfamily{#3}\fontseries{#4}\fontshape{#5}%
  \selectfont}%
\fi\endgroup%
\begin{picture}(3327,1218)(53,-403)
\put(2330,338){\makebox(0,0)[lb]{\smash{{\SetFigFont{12}{14.4}{\rmdefault}{\mddefault}{\updefault}{\color[rgb]{0,0,0}$C=(C_1,C_2)$}%
}}}}
\put(801,391){\makebox(0,0)[lb]{\smash{{\SetFigFont{12}{14.4}{\rmdefault}{\mddefault}{\updefault}{\color[rgb]{0,0,0}$C=(C_1,C_2)$}%
}}}}
\put(237,-207){\makebox(0,0)[lb]{\smash{{\SetFigFont{12}{14.4}{\rmdefault}{\mddefault}{\updefault}{\color[rgb]{0,0,0}$\X(C)$}%
}}}}
\put(1973,-233){\makebox(0,0)[lb]{\smash{{\SetFigFont{12}{14.4}{\rmdefault}{\mddefault}{\updefault}{\color[rgb]{0,0,0}$\X(C)$}%
}}}}
\end{picture}%

%% file: xc-arxiv-2.bbl
\begin{thebibliography}{10}
\providecommand{\url}[1]{{#1}}
\providecommand{\urlprefix}{URL }
\expandafter\ifx\csname urlstyle\endcsname\relax
  \providecommand{\doi}[1]{DOI~\discretionary{}{}{}#1}\else
  \providecommand{\doi}{DOI~\discretionary{}{}{}\begingroup
  \urlstyle{rm}\Url}\fi

\bibitem{acc07}
Acciaio, B.: Optimal risk sharing with non-monotone monetary functionals.
\newblock Finance and Stochastics \textbf{11}, 267--289 (2007)

\bibitem{acc:svi09}
Acciaio, B., Svindland, G.: Optimal risk sharing with different reference
  probabilities.
\newblock Insurance Math. Econom. \textbf{44}, 426--433 (2009)

\bibitem{amin:fil:min13}
Amini, H., Filipovi{\'c}, D., Minca, A.: Systemic risk with cetral counterparty
  clearing.
\newblock Research paper series 13-34, Swiss Finance Institute (2013)

\bibitem{Asim:Bad:Tsan13}
Asimit, A.V., Badescu, A.M., Tsanakas, A.: Optimal risk transfers in insurance
  groups.
\newblock Eur. Actuar. J. \textbf{3}, 159--190 (2013)

\bibitem{aub:fra90}
Aubin, J.P., Frankowska, H.: Set-Valued Analysis, \emph{System and Control,
  Foundation and Applications}, vol.~2.
\newblock Birkh{\"a}user, Boston (1990)

\bibitem{bar:elk05}
Barrieu, P., El~Karoui, N.: Inf-convolution of risk measures and optimal risk
  transfer.
\newblock Finance and Stochastics \textbf{9}, 269--298 (2005)

\bibitem{bar:sca08}
Barrieu, P., Scandolo, G.: General {Pareto} optimal allocations and
  applications to multi-period risks.
\newblock Astin. Bull. \textbf{38}, 105--136 (2008)

\bibitem{delb12}
Delbaen, F.: Monetary Utility Functions.
\newblock Osaka University Press, Osaka (2012)

\bibitem{fein:rud:web15}
Feinstein, Z., Rudloff, B., Weber, S.: Measures of systemic risk.
\newblock Tech. rep., arXiv:1502.07961v2 [q-fin.RM] (2015)

\bibitem{fil:kun08}
Filipovi{\'c}, D., Kunz, A.: Realizable group diversification effects.
\newblock Life \& Pensions pp. 33--40 (2008, May)

\bibitem{fil:kup07}
Filipovi{\'c}, D., Kupper, M.: On the group level {Swiss} {Solvency} {Test}.
\newblock Bulletin of the Swiss Association of Actuaries \textbf{1}, 97--115
  (2007)

\bibitem{fil:kup08}
Filipovi{\'c}, D., Kupper, M.: Optimal capital and risk transfers for group
  diversification.
\newblock Math. Finance \textbf{18}, 55--76 (2008)

\bibitem{fil:svi08}
Filipovi{\'c}, D., Svindland, G.: Optimal capital and risk allocations for law-
  and cash-invariant convex functions.
\newblock Finance and Stochastics \textbf{12}, 423--439 (2008)

\bibitem{ger79}
Gerber, H.U.: An Introduction to Mathematical Risk Theory, \emph{Huebner
  Foundation Monograph}, vol.~8.
\newblock Wharton School, University of Pennsylvania (1979)

\bibitem{ham:hey10}
Hamel, A.H., Heyde, F.: Duality for set-valued measures of risk.
\newblock SIAM J. Financial Math. \textbf{1}, 66--95 (2010)

\bibitem{ham:hey:rud11}
Hamel, A.H., Heyde, F., Rudloff, B.: Set-valued risk measures for conical
  market models.
\newblock Math. Finan. Economics \textbf{5}, 1--28 (2011)

\bibitem{ham:rud:yan13}
Hamel, A.H., Rudloff, B., Yankova, M.: Set-valued average value at risk and its
  computation.
\newblock Math. Finan. Economics \textbf{7}, 229--246 (2013)

\bibitem{jouin:sch:touz08}
Jouini, E., Schachermayer, W., Touzi, N.: Optimal risk sharing for law
  invariant monetary utility functions.
\newblock Math. Finance \textbf{18}, 269--292 (2008)

\bibitem{kab:saf09}
Kabanov, Y.M., Safarian, M.: Markets with Transaction Costs. Mathematical
  Theory.
\newblock Springer, Berlin (2009)

\bibitem{kain:rues09}
Kaina, M., R{\"u}schendorf, L.: On convex risk measures on {$L^p$}-spaces.
\newblock Math. Meth. Oper. Res. \textbf{69}, 475--495 (2009)

\bibitem{kel07}
Keller, P.: Group diversification.
\newblock Geneva Pap \textbf{38}, 382--392 (2007)

\bibitem{kis:rue10}
Kiesel, L., R{\"u}schendorf, L.: On optimal allocation of risk vectors.
\newblock Insurance Math. Econom. \textbf{47}, 167--175 (2010)

\bibitem{lep:mol16}
L\'epinette, E., Molchanov, I.: Risk arbitrage and hedging to acceptability.
\newblock Tech. rep., arXiv:1605.07884 [q-fin.RM] (2016)

\bibitem{lud07}
Luder, T.: Modelling of risks in insurance groups for the {Swiss} {Solvency}
  {Test}.
\newblock Bulletin of the Swiss Association of Actuaries \textbf{1}, 85--96
  (2007)

\bibitem{mas13}
Masayasu, K.: Insurance group risk management model for the next-generation
  solvency framework.
\newblock APJRI \textbf{7}, 27--52 (2013)

\bibitem{mo1}
Molchanov, I.: Theory of Random Sets.
\newblock Springer, London (2005)

\bibitem{cas:mol14}
Molchanov, I., Cascos, I.: Multivariate risk measures: a constructive approach
  based on selections.
\newblock Math. Finance \textbf{26}, 867–--900 (2016)

\bibitem{rav:svi14}
Ravanelli, C., Svindland, G.: Comonotone {Pareto} optimal allocations for law
  invariant robust utilities on {$L^1$}.
\newblock Finance and Stochastics \textbf{18}, 249--269 (2014)

\bibitem{wang14}
Wang, R.: Subadditivity and indivisibility (2014).
\newblock Preprint, University of Waterloo

\end{thebibliography}
